\documentclass{article}
\usepackage[utf8]{inputenc}

\usepackage{graphicx}
\usepackage{amsmath}
\usepackage{amsfonts}
\usepackage{amssymb}
\usepackage{amsthm}
\usepackage{enumitem}
\usepackage{breqn}
\usepackage{algorithm}
\usepackage{algpseudocode}
\usepackage[font=small,labelfont=bf]{caption}

\newtheorem{definition}{Definition}[section]

\newtheorem{lemma}[definition]{Lemma}
\newtheorem{theorem}[definition]{Theorem}
\newtheorem{corollary}[definition]{Corollary}
\usepackage{changepage}
\usepackage{mathtools}
\newcommand{\figureWidthAdjustment}{-2cm}
\newcommand\Set[2]{\{\,#1\mid#2\,\}}

\usepackage{multirow}
\usepackage{hhline}
\usepackage{makecell}
\usepackage[T1]{fontenc}
\usepackage[utf8]{inputenc}
\usepackage{xcolor}
\usepackage{array}
\usepackage{subcaption}
\usepackage{hyperref}

\begin{document}

\title{Stigmergy-based, Dual-Layer Coverage of Unknown Indoor Regions}
\author{Ori Rappel, Michael Amir, and Alfred M. Bruckstein}

\maketitle 

\begin{abstract}
We present algorithms for  uniformly covering an unknown indoor region with a swarm of simple, anonymous and autonomous mobile agents. The exploration of such regions is made difficult by the lack of a common global reference frame, severe degradation of radio-frequency communication, and ground obstacles. We propose addressing these challenges by using airborne agents, such as Micro Air Vehicles, in dual capacity, both as mobile explorers and, once they land, as beacons that help  other agents navigate the region. 

The algorithms we propose are designed for a swarm of identical ant-like agents with local sensing capabilities. The agents enter the region, which is represented as a graph, over time from one or more entry points and are required to occupy all of its vertices. Unlike many works in this area, we consider the task of informing an outside operator with limited information that the coverage mission is complete. Even with this additional requirement we show,  both through simulations and mathematical proofs, that the dual role concept results in linear-time termination, while also improving many well-known algorithms in the literature in terms of energy use.\footnote{This document is the full version of the corresponding AAMAS 2023 (International Conference on Autonomous Agents and Multiagent Systems) paper ``Stigmergy-based, Dual-Layer Coverage of Unknown Regions'', containing additional details and analysis.}

\end{abstract}




         
\newcommand{\BibTeX}{\rm B\kern-.05em{\sc i\kern-.025em b}\kern-.08em\TeX}


\section{Introduction}


Rescue workers coming to a rescue scene have neither an up-to-date map of the region nor a way of monitoring hazards such as harmful chemical leakages or structural instability. For the workers to safely and efficiently navigate the scene, it is desirable to continuously monitor it before and during the workers' operations. Several works have proposed addressing this problem by flooding the region with a swarm of simple, airborne agents prior to workers' entrance  \cite{howard_incremental_2002,hsiang_algorithms_2004,fekete_deployment_2008,hideg2022improved}. The robots' purpose is to explore and monitor the scene in a distributed fashion,  providing rescue workers with real time,  potentially life-saving information. 

Swarm-robotic algorithms whose purpose is to flood an unknown region with robots are broadly referred to as ``uniform dispersion'' algorithms in the literature \cite{hsiang_algorithms_2004}. Uniform dispersion algorithms are designed with an indoor, unknown rescue scene in mind. Such a setting poses  four unique challenges for multi-robot systems: (i) robots' movements may be hampered by rubble on the ground or other obstacles, (ii) lack of a common reference frame makes coordination based on common landmarks difficult or even impossible, (iii) robots have limited energy stores and can only operate at full capacity for a limited time, and (iv) in difficult indoor rescue regions, severe multi-path and signal attenuation errors limit radio frequency communication between the robots. Although numerous works have been written about each of these challenges as separate topics (see below the ``Related work'' section), we are not aware of any prior work about uniform dispersion that attempts to address all four challenges. 

In this work, we propose overcoming these challenges by using Micro Air Vehicles (e.g., small  quadcopters), which we will refer to as agents. We envision agents small enough to navigate difficult indoor environments (such as collapsed buildings or underground tunnels) and to settle inside. The agents enter the region one after another via predefined entry points and explore it in a distributed manner according to a local sensing-based algorithm. Agents gradually create \textit{beacons} throughout the region. These beacons help agents navigate  through local visual signals, and enable real time  monitoring and pathfinding inside the region.


This kind of airborne swarm approach is limited by several factors: first, using  airborne agents comes at the cost of limited flight time due to their high power consumption and limited battery capacity. Next, although rescue scenes such as collapsed buildings befit small agents, their limited payload capability severely limits their ability to carry and dispense the aforementioned beacons. To overcome both energy and payload limitations, rather than have beacons be something that the agents dispense throughout the region, our idea is to let the agents settle down as soon as possible and act as beacons themselves, so as both to  conserve energy and guide the remaining flying agents. This idea draws upon the bio-inspired concept of \textit{stigmergy}--communication via the environment \cite{zedadra2017multi,kennedy_swarm_2001,stirling_energy-efficient_2010}. In swarm robotics, stigmergy is usually implemented through agents \textit{deploying} beacons, ``virtual pheromones'', or other kinds of marks inside the environment. In contrast, our implementation of stigmergy utilizes \textit{agents themselves} as a medium to change the state of the environment by having agents \textit{become} beacons.

The algorithms we propose successfully cover regions that can be represented as a graph $\mathbb{G}(V,E)$, where vertices represent locations and edges between them represent the ability of an agent to move from one location to the next. For the sake of concreteness, we model the rescue scene as a discrete 2D grid graph composed of equally sized cells (Figure \ref{fig:simulationexampleregions}). We assume that each cell can contain at most two agents: one agent on the ground serving as a beacon, and one agent flying above-ground. Once an agent enters the region it moves autonomously according to a local sensing-based algorithm. As is common in swarm robotics, we assume the agents to be anonymous and identical, hence executing the same algorithm. The algorithm results in an  environment uniformly covered with beacons, such that each cell of the grid graph contains a beacon. We assume agents are capable of reading local visual signals from nearby beacons. Visual signals (i.e., lights)  can encode  information using color and blinking frequency, as shown in e.g. \cite{maxseiner2021visiblelights}. Since the agents’ direct (e.g., radio) communication capabilities are severely degraded,  they rely on these signals to navigate the environment. 

\textbf{Results.}\space\space We present several stigmergy-based ``gradient ascent''  algorithms for uniformly filling the region with beacons and study their time to completion and energy use via simulations and formal analysis. The algorithms work by maintaining and growing a subregion of beacons that generate a virtual ``gradient'' which flying agents ascend to efficiently find new locations to explore. We prove that the algorithms’ time to completion is linear in the size of the environment in both a synchronous and a bounded asynchronous time setting (in the sense of \cite{defago_fault-tolerant_2006}). The best-performing algorithm that we study is the ``Single-Layer Unlimited Gradient'' (SLUG) algorithm. We compare SLUG to recent works \cite{amir_fast_2019} as well as to the well-known, non-stigmergy based algorithms of Hsiang et al. \cite{hsiang_algorithms_2004} and find that SLUG's energy use is significantly lower, suggesting  that dual layer-based approaches to coverage have significant energy savings over traditional approaches. Desirably, we also show that SLUG can recover from several kinds of crashes and errors. 

An often overlooked topic in swarm robotics which is especially relevant in non-synchronous time settings (where we cannot infer much from the passage of time) is the ability of an operator to \textit{detect} mission completion. In this work, we take the  perspective that an external observer’s (e.g., a rescue worker's)  ability to tell whether the agents have completed the coverage mission is just as important as the agents’ ability to do so. Hence, in our view, the agents’ labor is not complete until they have managed to somehow deliver information about the job's completion to such an   observer. The primary metric we study and minimize is thus termination time - the earliest point in time in which robots can inform an external observer monitoring only the region's entry points that the region is completely covered with beacons. We show that the observer can quickly and correctly be informed of mission termination through back-propagation of the beacons' states. 

\textbf{Related work.} 
Many works have been written on the topic of multi-agent deployment under constraints  \cite{peleg_distributed_2005,altshuler2018introduction,cortes_coverage_2004,optimal_physical_sorting,tran2021robust}. The problem of completely filling an environment with robots in a decentralized fashion is often called ``uniform dispersion.'' The version of uniform dispersion we consider in this work can be traced back to Hsiang et al. \cite{hsiang_algorithms_2004} or to Howard et al. \cite{howard_incremental_2002,howard_mobile_2002}. The literature on uniform dispersion and related problems can be organized according to assumptions regarding the capabilities of the agents, the time scheme used (synchronous or asynchronous), the initial location of the agents, and the properties of the region. As an example,  in \cite{hsiang_algorithms_2004}, agents are assumed capable of direct (radio-based) inter-agent communication, time is synchronous, agents are initially located outside the region, and the region is grid-like.  Conversely, in some works the initial location of the agents may be inside the region \cite{barriere_uniform_2011,howard_incremental_2002}.  The region to be covered may be a continuous part of the two-dimensional plane  \cite{stirling_energy-efficient_2010,howard_incremental_2002,mulgaonkar_robust_2018,rebuttal_icracoverage_ozdemir2019spatial} or modelled as a discrete graph (usually as a grid graph) \cite{hsiang_algorithms_2004,fekete_deployment_2008,flocchini_uniform_2014,barriere_uniform_2011,hideg2020asynchronous,hideg2022improved}. 

In this work we broadly follow Hsiang et al. \cite{hsiang_algorithms_2004}'s setting and assume that our agents are initially located outside the region, gradually entering it via predefined \textit{entry points}. However, whereas Hsiang et al.'s algorithms require synchronous time and direct inter-agent messaging, our algorithms only require indirect visual communication and can operate in a bounded asynchronous time scheme. Additionally, unlike \cite{hsiang_algorithms_2004} and most works in this area, we focus on external detection of coverage completion rather than just achieving coverage. Hsiang et al.'s problem setting and model are extended in numerous works including \cite{fekete_deployment_2008,flocchini_uniform_2014,barriere_uniform_2011,hideg2020asynchronous,hideg2022improved,amir_fast_2019,amir_minimizing_2019}. The closest work to ours among these is 
\cite{amir_fast_2019}, in which the concept of covering an environment with two layers of agents is similarly explored. The coverage algorithm presented in \cite{amir_fast_2019} (``DLTT'') is used as a benchmark in our simulation results. We show that SLUG  beats it significantly in terms of termination time and energy use.

Stigmergy-based graph exploration strategies in which agents place beacons in the environment to aid the exploration process are widely studied \cite{rebuttal1disser2019tight,rebuttal2funabiki2020range}. In contrast to the majority of these works, our robots do not deploy beacons in the environment, but \textit{become} them. We show that using \textit{agents themselves} as beacons results in highly attractive properties, including significant energy savings and ant-like recovery from crashes and errors. The idea of using settled airborne agents as beacons was first studied by Stirling et al. \cite{stirling_energy-efficient_2010,martinoli_energy-time_2013} and later by Aznar et al. \cite{aznar_modelling_2018,aznar_energy-efficient_2018} with the objective of exploring (not coverage) and in \cite{amir_fast_2019,rappel_area_2019} with the objective of coverage.  In \cite{rappel_area_2019} the authors present a uniform coverage algorithm based on stigmergy and depth-first search, however the number of agents in each cell is not limited and agents possess unique IDs. The assumption of unique IDs is made in \cite{stirling_energy-efficient_2010, aznar_gregori_uav_2020,martinoli_energy-time_2013,aznar_modelling_2018} as well. Conversely, in this work and in \cite{amir_fast_2019} (whose DLTT algorithm we improve upon, as per the previous paragraph) it is assumed all robots in the swarm are anonymous and identical. 

Energy is the primary constraint on the performance of swarms composed of flying agents, and energy consumption is often measured as a performance metric \cite{stirling_energy-efficient_2010,amir_minimizing_2019,rappel_area_2019,aznar_modelling_2018}. In \cite{amir_minimizing_2019}, a time and energy-optimal uniform dispersion algorithm for simply connected environments is presented alongside theoretical lower bounds on energy use in general environments. Aznar et al. \cite{aznar_modelling_2018} propose predicting energy consumption via recurrence equations.  We study our algorithms' energy use in Section \ref{section:simulations}, measuring both the maximum individual energy use of an agent and the cumulative energy use of the swarm. Our thesis, which is corroborated by simulation results (Section  \ref{section:simulations}), is that a stigmergy-based approach where agents can land and become beacons leads to significant energy savings.

\vspace{-2mm}
\section{Model and Problem Formulation}

In this section we define the \textit{beacon coverage problem}, a type of uniform dispersion task  wherein a swarm of agents seeks to completely cover a discretized, a priori unknown environment $\textbf{R}$ with beacons and to detect termination, such that in the final configuration every location in $\textbf{R}$ contains a beacon and an external observer monitoring only the state of the agents at the region's entry points is capable of detecting that the coverage mission is complete. 

We model $\textbf{R}$ as a graph $\mathbb{G}(V,E)$ whose vertices, $V$, represent locations in $\textbf{R}$ and whose edges, $\textbf{E}$, represent the connections between locations. We define $n = |V|$ and $m = |E|$. For simplicity, we assume here that $\mathbb{G}$ is an arbitrary grid graph composed of equally sized cells. However, this assumption is not required by our algorithms and formal analysis, both of which can readily be extended to arbitrary graph environments.

We assume a sufficiently large swarm of identical and anonymous robots $a_1, a_2, \ldots$ is ready to enter the region one by one via  predefined entry points  located somewhere in  $\textbf{R}$. $a_i$ is defined to be the $i$th robot that enters  $\textbf{R}$.

\subsection{The Agents}

At all times, an agent is either in flight and searching for a place to settle, or has already settled somewhere in $\textbf{R}$. The former type of agent is called \textit{mobile} whereas the latter is called a \textit{beacon}.  Beacons have two sub-states: they can be \textit{open} or \textit{closed}. Beacons remain in place forever, conserving energy, and their function is to provide simple information to nearby mobile agents via visual signals.

Agents are assumed to be identical (they all act according to  the same algorithm), anonymous (have no identifiers) and autonomous. Agents wake up once per time step and perform a Look, Compute and Move action cycle. An agent can sense other agents' presence at any vertex neighboring its location. As we assume $\textbf{R}$ is a grid graph, this means agents sense adjacent cells at Manhattan distance $1$ from themselves. Mobile agents are allowed to move to an adjacent cell as long as that cell contains no other mobile agent. Mobile agents can settle and become beacons at any adjacent cell that contains no other beacons. Hence, each cell of $\textbf{R}$ may contain at most one mobile agent and one beacon, reflecting the idea that beacons have landed on the ground and mobile agents can fly above them.

\begin{definition}
The set of agents and cells sensed by agent $a_i$ at time $t$, denoted $\xi(i,t)$, is called ``the neighborhood of $a_i$ at time $t$.''
\end{definition}

All agents are capable of detecting the state of agents in their neighborhood (mobile, open beacon, closed beacon). Each agent $i$ further keeps track of an integer called a step count $s_i$. Beacons are capable of projecting their step counts to neighboring agents.

\subsection{Time Model}
\label{section:timescheme}

The algorithms we describe work in many kinds of synchronous, semi-synchronous and asynchronous time settings. For the purposes of formal modelling and time bound analysis, in this paper we focus on a type of fair 2-bounded asynchronous scheduler (in the sense of \cite{defago_fault-tolerant_2006}).  Time, $t$, is divided into integer time steps $t = 0,1,2,3,\ldots$, such that each time step is further sub-divided into $M$ equal sub-time steps of duration $dt=1/M$. Let $U(0,M-1)$ denote the uniform distribution over the integers $\{1,2,\ldots M-1\}$, and  let $t_{i,k}$ denote the $k$th time that agent $a_i$ wakes up. We define:

\begin{equation}
    t_{i,k} = k + \frac{U(0,M-1)}{M}
\end{equation}

In other words, at every time step, each agent wakes up during a randomly chosen sub-time step (see Figure \ref{fig:timescheme}). An agent that wakes up at time $t$ senses the configuration of agents at the previous sub time step $t-dt$. If two agents attempt to move to the same location at the same sub time-step, we assume one of them stays put arbitrarily.
This kind of time scheme is applicable to scenarios in which agents have an internal timer that wakes them up at some fixed nominal rate, but are not controlled by a centralized scheduler that guarantees precise synchronization between the agents’ clocks. 

\begin{figure}
    \centering
    \includegraphics[width=120mm]{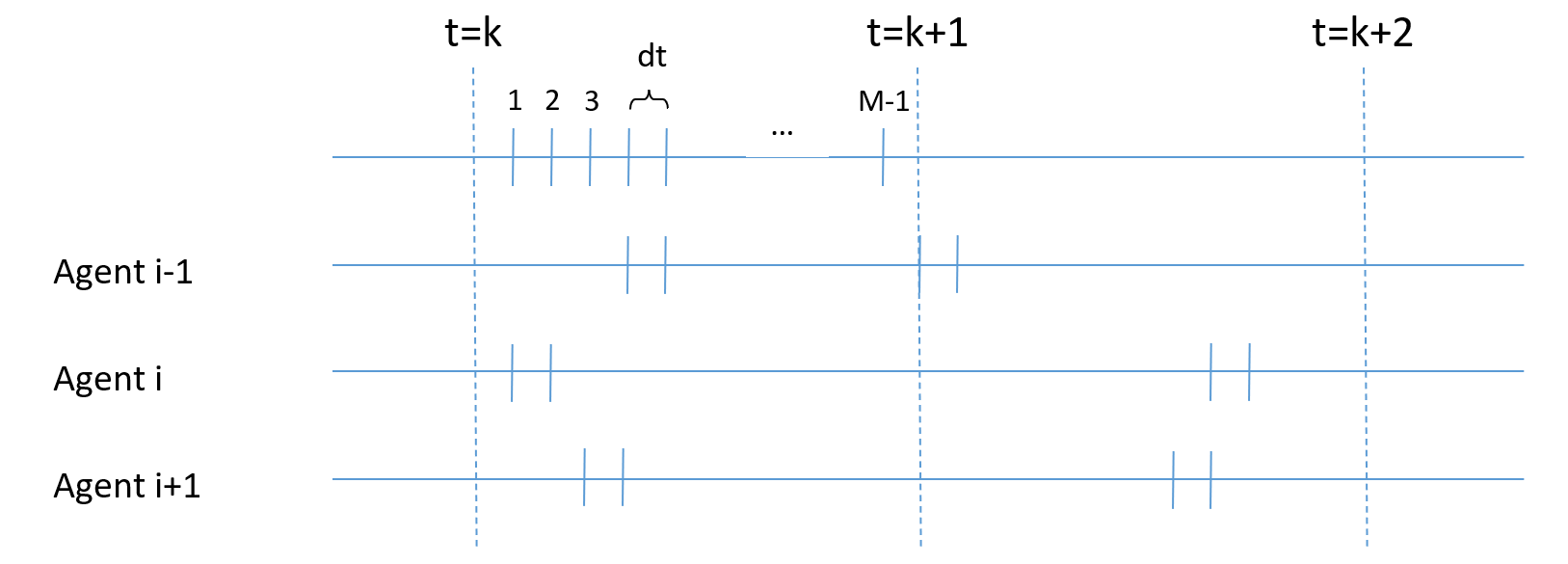}
    \vspace{-2mm}
    \caption{An illustration of the time scheme. Large vertical bars indicate the beginning of time steps; the short bars indicate sub-time steps. The randomly chosen wake-up times of agents $i-1$, $i$ and $i+1$ are depicted.}
    \vspace{-3mm}
    \label{fig:timescheme}
\end{figure}

For the rest of this work, unless stated otherwise, the phrase ``at time step $t$'' refers to the time interval $[t,t+1)$, and the phrase ``at time $t$'' refers to the configuration of the agents at sub-time step $t$, \textit{after} any agent wake-ups.

As a special case, our time bound result   (Theorem \ref{theorem:dllgterminationtime}) holds  also for synchronous time schemes where all agents always wake up at the same sub-time step.


\vspace{-2mm}
\subsection{Agent Entry Model}
\label{section:entrymodel}

Agents are initially outside the region and enter the region one by one via predefined \textit{entry points}. We assume a constant rate entry model in which the time between successive entry attempts of agents to the region through an entry point is bounded by an integer parameter  $\Delta T  \geq 1$. Specifically, let $p$ be an entry point in $\textbf{R}$. For every integer $k \geq 0$, an agent will enter $p$ at the earliest sub-time step $t_0 \in [k \Delta T, (k+1) \Delta T)$ at which $p$ contains no mobile agent. Since $t_0$ is unique by definition (there is just one ``earliest'' time $t \in [k \Delta T, (k+1) \Delta T)$ where $p$ is empty), at most $k$ agents can enter the region before time $k\Delta T$.  An agent entering the environment at a given time step remains there without waking for the rest of the current time step and activates only at the subsequent time step.

Although many works in the literature of robotic coverage assume similar entry models, the effect of the entry process on mission completion time is rarely explored in depth. The interaction between the entry delay $\Delta T$ and the coverage algorithm is especially important when trying to develop a tight upper bound on the time to mission completion. We study this topic in Sections  \ref{section:DLLG} (analytically) and \ref{section:simulations} (empirically). 


\subsection{The Agents' Mission}
\label{section:objective}

It is common in works on multi-robot coverage to place primary importance on the amount of time a coverage algorithm takes to complete. One limitation of this approach is that an external observer cannot necessarily tell that coverage has been completed, and might not be able to proceed to the next stage of their operation without such information. In this work, we take the perspective that an external observer’s ability to tell whether the agents have arrived at the desired configuration is just as important as the agents’ ability to do so, and so the agents’ mission is not complete until they have managed to somehow deliver this information to such an   observer. In light of this, our agents’ mission is both (a) to completely cover $\textbf{R}$ with beacons and (b) to deliver this information back to external observers. Since $\textbf{R}$ is an a priori unknown region, we assume external observers can only see and read the states of the agents at each of $\textbf{R}$’s entry points. We then follow the convention that, to announce mission completion, each entry point must contain a closed beacon. Formally:

\begin{definition}
The \textbf{termination time}, $T_C (\textbf{R})$, of the multi-agent system  with respect to an environment $\textbf{R}$ is the first time $t$ at which every entry point contains a closed beacon. The algorithm is said to have executed correctly if every cell contains an (open or closed) beacon at termination time. The \textbf{termination time step} is  defined as the first time step after termination occurs, i.e., $\lceil T_C (\textbf{R}) \rceil$.
\label{def:terminationtime}
\end{definition}

We are also interested in the \textit{total energy use} of a multi-agent system as a measure of its efficiency. We assume here that the power consumption of a mobile agent is constant and that once an agent becomes a beacon its energy consumption is 0 (since it no longer needs to fly). Hence energy consumption is linear with respect to the time an agent is mobile, and the total energy used is the sum of the flight times of all the agents:

\begin{definition}
The energy consumption of an agent $a_i$, denoted $E_i$, is the duration of time $a_i$ spends as a mobile agent after entering $\textbf{R}$ up to  time $T_C(\textbf{R})$.
\end{definition}

\begin{definition}
\label{def:totalenergy}
Assuming $k$ agents entered $\textbf{R}$ before time $T_C(\textbf{R})$, the swarm's total energy consumption is equal  to $\sum_{i=0}^k E_i$.  Total energy consumption is  denoted by $E_{total}(\textbf{R})$.
\end{definition}

We may also consider the \textit{maximal energy use} of an agent in the system. Whereas total energy measures the swarm's cost-effectiveness (e.g., how much fuel or electricity it consumes), maximum energy tells us what energy storage capabilities our robots require to complete the mission:

\begin{definition}
\label{def:maximumenergy}
The maximum energy consumption of a multi-agent system with respect to an environment $\textbf{R}$ is equal to $\max  E_i$, the maximum energy used among all agents at termination time.
\end{definition}

We note that even after the termination time, redundant mobile agents may remain in $\textbf{R}$ and continue consuming energy. Although post-termination energy use is not one of the main metrics we explore in this paper, we will later show that beacons can be used to send redundant agents back to entry points where they can be recovered and stop consuming energy.

\section{Algorithms}

In this section we describe two beacon coverage algorithms: ``Dual-Layer Limited Gradient,'' abbreviated \textit{DLLG}, and ``Single-Layer Unlimited Gradient,'' abbreviated \textit{SLUG}. DLLG is a relatively straightforward algorithm that completely fills the region with two layers of agents: beacons and mobile agents. We study its properties and prove that, when there is one entry point, its termination time step obeys the inequality  $ (2n-1)\Delta T+1 \leq \lceil T_C (\textbf{R}) \rceil \leq (2n-1)\Delta T+2$ for  $\Delta T \geq 2$. We then discuss SLUG, which is a strict improvement over DLLG and uses DLLG as a building block. SLUG improves DLLG in two ways: by making better use of closed beacons, and by enabling greater freedom of movement for the robots. These improvements can be implemented separately. In Section \ref{section:simulations}, we study (numerically) the contribution of each type of improvement to the algorithm's performance in terms of termination time and energy use. 

In addition to its performance improvements over DLLG, SLUG can recover from various errors and crash faults - see Section \ref{section:crashfaults}.

\subsection{Dual-Layer Limited Gradient}
\label{section:DLLG} 

The idea of DLLG is to have mobile agents read the step counts of neighboring beacons and treat them as a gradient which they must ascend. Specifically, at every time step, mobile agents with step count $s_i$ attempt to move to neighboring locations that contain a beacon with step count $s_i+1$. Hence, the step count of every mobile agent necessarily equals the step count of the beacon it flies over, and agents are always “climbing” the gradient defined by the beacons in single-step increments. Mobile agents themselves become beacons over time: as a first priority, whenever a mobile agent sees a neighboring empty cell, it attempts to settle in that cell, changing its state to “open beacon” and thereafter signaling its step count to any agent that sees it. Beacons become “closed” in a backpropagating fashion until each entry point contains a closed beacon. Specifically, open beacons become “closed” when they contain a mobile agent, all their neighboring cells contain beacons, and every neighboring beacon with higher step count (if any such neighbor exists) is a closed beacon. This 
condition guarantees that the algorithm terminates when precisely $2n$ agents are present in the region: $n$ mobile agents and $n$ closed beacons.

DLLG is a distributed algorithm that is implemented via a local action rule that all agents in $\textbf{R}$ execute at every time step. The local rule for agent $a_i$ at time $t$ is defined by Algorithm \ref{algorithm:DLLG}. Lines 1-9 of Algorithm \ref{algorithm:DLLG} govern mobile agents' behavior and we shall refer to them as the \textit{mobile agent action rules}. Lines 10-11 govern the behavior of beacons and shall be referred to as \textit{beacon action rules}. In DLLG, beacon action rules are only necessary for detecting termination; the mobile agent action rules are entirely sufficient for the coverage process itself.

\begin{algorithm}[!ht]
\caption{Local rule for agent $a_i$ waking up at time $t+dt$ and executing Dual-Layer Limited Gradient.}\label{algorithm:DLLG}
\begin{algorithmic}[1]
\Require $a_i$ has entered $\textbf{R}$
\If{$a_i$ is a mobile agent}
\If{$\xi(a_i,t)$ contains an empty cell $u$}
\State $a_i$ moves to $u$ 
\State $a_i$ sets state to ``open beacon''
\State $s_i \gets s_i + 1$ 
\ElsIf{$\xi(a_i,t)$ contains a cell $u$ with exactly one agent $a_j$ \textbf{and} $a_j$ is an open or closed beacon \textbf{and} $s_j = s_i + 1$}
\State $a_i$ moves to $u$
\State $s_i \gets s_i + 1$ 
\EndIf
\ElsIf{$a_i$ is an open beacon \textbf{and} there is a mobile agent at $a_i$'s location \textbf{and} every beacon $a_j \in \xi(a_i,t)$ with $s_j > s_i$ is closed \textbf{and} $\xi(a_i,t)$ does not contain an empty cell}
\State $a_i$ sets state to ``closed beacon''
\EndIf
\end{algorithmic}
\end{algorithm}


We can think of DLLG as creating a directed acyclic graph $\mathcal{D}$ whose vertices are beacons and where there is an edge from a beacon $a_i$ to any neighboring beacon whose step count is $s_i+1$ (Definition \ref{def:DacyclicDLLG}). Mobile agents are only allowed to move via directed edges of this DAG, and they follow the paths induced by its edges until the region is completely filled.

\subsubsection*{Analysis.}

The rest of this section concerns the mathematical analysis of DLLG. We prove that when $\Delta T \geq 2$, the termination time step $\lceil T_C (\textbf{R}) \rceil$ of DLLG over any environment with one entry point obeys the inequality $ (2n-1)\Delta T+1 \leq \lceil T_C (\textbf{R}) \rceil \leq (2n-1)\Delta T+2$ (Theorem \ref{theorem:dllgterminationtime}). This bound on termination  is deterministically true regardless of the randomly determined  wake-up times of the agents in a given run of DLLG, as long as we agents obey our bounded asynchronous time scheme; in particular, it remains true if we assume all agents always wake up at the same sub-time step, i.e., remains true in a synchronous time scheme.

\begin{theorem}
If $\Delta T \geq 2$, the termination time step $\lceil T_C (\textbf{R}) \rceil$ of DLLG  over any environment $\textbf{R}$ with $n$ cells and one entry point obeys the inequality $ (2n-1)\Delta T+1 \leq \lceil T_C (\textbf{R}) \rceil \leq (2n-1)\Delta T+2$.
\label{theorem:dllgterminationtime}
\end{theorem}

When there are multiple entry points, or when $\Delta T = 1$, intuition and countless numerical simulations suggest that the termination time of any execution of DLLG is bounded by, but can be lower than $(2n-1)\Delta T+2$. The formal time bound analysis of these two situations, and especially the $\Delta T = 1$ case, is (surprisingly) much more complex and is left as an open problem.

For the rest of this section, let $p$ be the lone entry point of $\textbf{R}$. As mentioned earlier in this section,  DLLG  implicitly generates a directed acyclic graph $\mathcal{D}$ which spans $\textbf{R}$. Let us formally define $\mathcal{D}$.

\begin{definition}
\label{def:DacyclicDLLG}
In a given execution of DLLG, define $\mathcal{D}(t)$ to be the subgraph of $\textbf{R}$ rooted at $p$ whose vertices are all the vertices of $\textbf{R}$ containing beacons at time $t$, and where there is an edge $(u,v)$ if $u$ contains a beacon $a_i$ and $v$ contains a beacon $a_j$ such that $s_j=s_i+1$. We define $\mathcal{D}=\mathcal{D}(\infty)$.
\end{definition}

We note that $\mathcal{D}$ might not be the same between different runs of DLLG—its precise structure depends on the stochastic activation times of the agents, as well as their sometimes arbitrary movement decisions. Since it is simpler to think in deterministic terms, for the rest of this section we will be studying a fixed, arbitrary execution of DLLG over $\mathcal{R}$, thus we may assume $\mathcal{D}$ is fixed.

If $\mathcal{D}(t)$ does not contain all cells in $\textbf{R}$, then there is at time $t$ a beacon $a_i$ which is adjacent to an empty cell $u$. There is necessarily also a path in $\mathcal{D}(t)$  from $p$ to $a_i$. DLLG  constantly sends agents across the edges of $\mathcal{D}(t)$, thus eventually an agent will arrive at $u$. This argument shows that $\mathcal{D}(t)$ grows with $t$ until there is a beacon at every cell of $\textbf{R}$. Hence $\mathcal{D}$  contains every vertex of $\textbf{R}$.

To prove Theorem 1 we first need to study the structure of $\mathcal{D}$. We make use of the following definitions:

\begin{definition}
The \textbf{children} of a vertex $v \in \mathcal{D}$ are all vertices reachable from $v$ via a directed path of $\mathcal{D}$. A \textbf{leaf} of $\mathcal{D}$ is a vertex $v \in \mathcal{D}$ with no children.
\end{definition}

\begin{definition}
Let $v$ be a vertex of $\mathcal{D}$ and let $a_i$ be the  beacon at $v$. Let $t$ be the sub-time step in which $a_i$ becomes a beacon. The depth of $v$, denoted $depth(v)$, is the step count $s_i$ of $a_i$ at time $t$ (e.g., $depth(p) = 0$, since the beacon at $p$ has step count $0$).
\label{def:depth_of_vertex}
\end{definition}

\begin{definition}
$dist(p,v)$ is the number of edges of the shortest directed path in $\mathcal{D}$ from $p$ to $v$. 
\end{definition}

\begin{lemma}
$dist(p,v) = depth(v)$
\label{lemma:distptov_equals_depthofv}
\end{lemma}

\begin{proof}
By definition, the directed edges of $\mathcal{D}$ determine which vertices a mobile agent can move to from a given location. In the DLLG algorithm, a mobile agent whose step count is $s$ can only move to vertices of depth $s+1$. Hence, an agent can only move from $p$ to $v$ in exactly $depth(v)$ steps, which means $dist(p,v)$ must equal $depth(v)$. 
\end{proof}

We define next the  ``depth'' of an agent (Definition \ref{def:agentdepth}), a time-dependent value associated with the agent. At every time step $t$, we keep track of the depth of both agents inside $\textbf{R}$ and agents outside $\textbf{R}$ that will enter $\textbf{R}$ in the future. When $a_i$ is inside $\textbf{R}$, we shall show in Lemma \ref{lemma:inductiondllg}  that its depth is related to the depth of the vertex $v$ it is located at (Definition \ref{def:depth_of_vertex}). 

\begin{definition}
\label{def:agentdepth}
Let $W_i(t) \coloneqq \Set{t'}{0 \leq t' \leq t \land \text{$a_i$ wakes up at time $t'$}}$ be the set of wake-up times of agent $a_i$ before time $t$ (inclusive), and define $W_i^*(t) \coloneqq \Set{t' \in \mathbb{Z}}{0 \leq t' \leq t \land \text{$a_i$ has not entered \textbf{R} at time $t'-dt$}}$ to be the set of integer times that $a_i$ begins outside $\textbf{R}$. 

The depth of an agent $a_i$ at time $t$ is defined as follows:  \begin{equation*}depth(a_i,t) \coloneqq |W_i(t)| + |W_i^*(t)|-(i-1)\Delta T\end{equation*}
\end{definition}

The expression $ |W_i(t)| + |W_i^*(t)|$ in Definition \ref{def:agentdepth} counts the number of times $a_i$ has woken up by time $t$, treating the integer times where $a_i$ is outside $\textbf{R}$ as ``virtual'' wake-ups. Every agent increases its depth  exactly once per time step during some sub-time step (note that this is not the same as moving to a vertex of greater vertex depth). Example:  suppose that $\Delta T = 2$, that the agent $a_3$ enters $\textbf{R}$ at time $t = 4$, and that its first wake-up time occurs at time $5+3dt$. Then $depth(a_3, 2) = -2$,  $depth(a_3, 3) = -1$, $depth(a_3, 4) = 0$, $depth(a_3, 5+dt) = 0$, $depth(a_3, 5+3dt) = 1$ (see Section \ref{section:timescheme} for the definition of $dt$). In this example $a_3$ enters $\textbf{R}$ exactly when its depth is $0$. In Lemma \ref{lemma:inductiondllg} we prove this must always be the case.

\begin{definition}
The vertex location of the agent $a_i$ at time $t$ is denoted $v(a_i,t)$.
\end{definition}

\begin{lemma}
\label{lemma:inductiondllg}
Assuming $\Delta T \geq 2$, at any  time $0 \leq t \leq  T_C(\textbf{R})$ the following holds for any vertex $v \in \mathcal{D}$ and any mobile agent $a_i$:
\begin{enumerate}
    \item If $(i-1)\Delta T \leq  t  $, $a_i$ has entered $\textbf{R}$ at time  $(i-1)\Delta T$ (which is the time $a_i$'s depth equals $0$).
    \item If $v(a_i,t)$ contains an open beacon then $depth(a_i,t) =depth(v(a_i,t))$ \textbf{or} the beacon at $v(a_i,t)$ becomes a closed beacon by time $\lceil t \rceil - dt$.
    \item 	If $v$ contains a closed beacon, it contains a mobile agent and all its children are closed beacons.
\end{enumerate}
\end{lemma}

\begin{proof}
We shall prove the lemma by induction on $t$. The base case is straightforward: at any time in $[0,\Delta T)$, there is just one agent at $p$ with depth $0$, thus statements (1)-(3) hold. Now assume for induction that at times $0,1, \ldots,t$, statements (1)-(3) are true. We shall show that that (1)-(3) are true at time $t+dt$ as well.

We first show that (3) must be true at time $t+dt$. Let $v \in \mathcal{D}$ be some vertex containing a closed beacon at time $t+dt$. Note that this means $v$ must contain either a closed or open beacon at time $t$, since open beacons must wake up at least once to become closed. If $v$ contains a closed beacon at time $t$, then by (3), it and all its children contain mobile agents at time $t$. Clearly, none of these mobile agents can ever move at any time after $t$, which means (3) remains true for $v$ at time $t+dt$.  If $v$ contains an open beacon at time $t$, then that beacon becomes a closed beacon at time $t+dt$. According to the DLLG algorithm, this can only occur if $v$ contains a mobile agent at time $t$ and all of $v$'s children are closed beacons at time $t$. By inductive assumption (3), this means all of $v$'s children also contain mobile agents. Hence none of these agents can ever move after time $t$, meaning that also in this case, (3) holds for $v$ at time $t+dt$.

This shows that (3) holds at time $t+dt$. We now need to show (1) and (2) hold for every agent $a_i$ at time $t+dt$. We separate our analysis into cases:

\begin{enumerate}
    \item $a_i$ entered $\textbf{R}$ at time $t+dt$.
    \item $a_i$ changed its location at time $t+dt$. 
    \item $a_i$ did not enter $\textbf{R}$ and did not  change its location at time $t+dt$.
\end{enumerate}

Note that if the entry point $p$ contains a closed beacon at time $t+dt$ then by definition $t+dt \geq T_C(\textbf{R})$, hence we may assume without loss of generality that $p$ contains an open beacon at time $t$.

\textbf{1st case:} Let us assume that at time $t+dt$ the agent $a_i$ entered $\textbf{R}$, and show that (1) and (2) must hold for this agent. By our inductive assumptions, $a_i$ entering $\textbf{R}$ at time $t+dt$ means that the first $i-1$ agents have already entered $\textbf{R}$ at times $\Delta T, 2 \Delta T \ldots (i-2) \Delta T$ respectively. This means that the earliest point in time $a_i$ can enter $\textbf{R}$ is $(i-1)\Delta T$, implying $t+dt \geq (i-1)\Delta T$. It remains to show that $t+dt = (i-1)\Delta T$. Assume for contradiction that $t+dt > (i-1)\Delta T$, then at time $(i-1)\Delta T - dt$ the agent $a_{i-1}$ is necessarily located at $p$, which is what prevents $a_i$ from entering at time $(i-1)\Delta T$. Since $a_{i-1}$ enters $p$ at time $(i-2)\Delta T$ and is still there at time $(i-1)\Delta T - dt$, and since $\Delta T \geq 2$, $a_{i-1}$ must be located at $p$ at time $(i-2)\Delta T + 2-dt$. But   $a_{i-1}$ wakes up exactly once during time step $(i-2)\Delta T + 1$, meaning  $depth(a_{i-1},(i-2)\Delta T + 2 - dt) = 1 \neq depth(p) = 0$ which contradicts inductive assumption (2). We arrived at a contradiction, therefore necessarily $t + dt = (i-1)\Delta T$. This establishes (1) for agent $a_i$ at time $t+dt$. It also establishes (2) for $a_i$ at time $t+dt$, since $depth(a_i,(i-1)\Delta T)=0$ and $depth(p) = 0$.

\textbf{2nd case:} Assume that $a_i$ changed its location at time $t+dt$. This implies that $a_i$ did not enter $p$ at time $t+dt$, since in our model, an agent cannot both enter $\textbf{R}$ and move in one time step. Hence  (1) remains true vacuously for $a_i$. Furthermore, because $a_i$ moved at time $t+dt$, it increased its depth by $1$ and its step count by $1$, meaning (2) must remain true for $a_i$ at time $t+dt$.

\textbf{3rd case:} Assume $a_i$ did not enter $\textbf{R}$ and did not change its location at time $t+dt$. (1) remains true vacuously in this case. If   $a_i$ does not wake up at time $t+dt$, (2) clearly remains true, so let us assume $a_i$ wakes at time $t+dt$. Denote  $t' = \lfloor t+dt \rfloor - dt$. Our inductive assumption says (2) is true at time $t'$, which implies that at time $t'$  all mobile agents $a_j, j \neq i$ for which $v(a_j,t')$ contains an open beacon have vertex depth difference at least $2$: 
\begin{equation}|depth(v(a_j,t')) - depth(v(a_i,t'))| = |depth(a_j,t') - depth(a_i,t')| \geq \Delta T \geq 2\tag{*}\label{eq:depthdifferencelemmaproof}\end{equation} Moreover, since we assume $a_i$ does not move at time step $t+dt$, we know all immediate children of $v(a_i,t+dt)$ (if any exist) contain mobile agents at time $t$. These mobile agents cannot have been located at open beacons at time $t'$ (because they cannot reach a child of $v(a_i,t+dt)$ in one move given (\ref{eq:depthdifferencelemmaproof})), so they must have been located at closed beacons at time $t'$, hence also at time $t$ (since, by (3), agents at closed beacons can never move). Hence, all  children of $v(a_i,t+dt)$ necessarily contain closed beacons at time $t'$. Consequently, by DLLG's beacon action rule, the beacon at $v(a_i,t+dt)$ will wake up and become closed by the end of the current time,  and (2) holds for $a_i$.  \end{proof}


\begin{corollary}
\label{corollary:DLLGagentrule}
DLLG fills the graph with $2n$ agents, such that the last agent always enters at time $(2n-1)\Delta T$.
\end{corollary}

Corollary \ref{corollary:DLLGagentrule} follows immediately from Lemma  $\ref{lemma:inductiondllg}$, (1). We now prove Theorem \ref{theorem:dllgterminationtime}.


\begin{proof}
Let us first establish the lower bound. Lemma \ref{lemma:inductiondllg}, (3) implies that termination can only occur when $\textbf{R}$ contains $2n$ agents ($n$ mobile and $n$ closed beacons).  Corollary \ref{corollary:DLLGagentrule} says that agent $a_{2n}$ enters at time $(2n-1)\Delta T$. The beacon at  $p$ can only become closed once it wakes up and sees agent $a_{2n}$, hence $\lceil T_C{\textbf{R}} \rceil \geq (2n-1)\Delta T + 1$.

It remains to establish the upper bound. When $a_{2n}$ enters $\textbf{R}$ there are $n$ mobile agents in $\textbf{R}$. By time $(2n-1)\Delta T+2-dt$, $a_{2n}$ has woken up at least once, hence its depth is $1$, but it could not have moved because every vertex in the graph contains a mobile agent. Hence by Lemma \ref{lemma:inductiondllg}, (2), $p$ is closed at time $(2n-1)\Delta T+2-dt$, establishing $\lceil T_C{\textbf{R}} \rceil \leq (2n-1)\Delta T + 2$. \end{proof}
\vspace{-2mm}

The bound of Theorem \ref{theorem:dllgterminationtime} is precise. Whether the termination time step is $(2n-1)\Delta T + 1 $ or $(2n-1)\Delta T + 2 $ depends on the wake-up times of the agents. The upper bound $(2n-1)\Delta T + 2 $   can for example be attained  on a region that is a linear path of $n$ vertices (Figure \ref{fig:simulationexampleregions}, (a)), by always waking up the beacons before the mobile agents at each time step; the lower bound $(2n-1)\Delta T + 1 $ is attained on the same region by always waking the beacons after the mobile agents.


\subsection{Single-Layer Unlimited Gradient}

In the previous section we studied the DLLG algorithm, which terminates with $2n$ agents in $\textbf{R}$ - $n$ agents and $n$ beacons. DLLG is limited by two factors:

\begin{enumerate}
    \item Agents can traverse only the edges of the directed acyclic graph $\mathcal{D}$, but oftentimes there are edges in $\textbf{R}$ that are not in $\mathcal{D}$ which could shorten the agents' paths.
    \item The algorithm requires $2n$ agents to be present in $\textbf{R}$ before termination can be detected, but only $n$ agents are necessary to cover $\textbf{R}$.
\end{enumerate}

The Single-Layer Unlimited Gradient algorithm (Algorithm \ref{algorithm:SLUG}) addresses these weaknesses through three straightforward improvements to DLLG. The first improvement SLUG  makes is using ``looser'' (i.e., less constrained) agent traversal rules. In DLLG, a mobile agent $a_i$ can only move to a beacon whose step count is $s_i+1$. In  SLUG, $a_i$ can move to any beacon whose step count is greater than $s_i$ (if $a_i$’s step count is $5$, it can move to a beacon with step count $6$, or $10$, or $100$, etc.). Under this \textit{``unlimited gradient''} traversal rule, the agents still implicitly generate and traverse a DAG $\mathcal{D}$ that spans the cells of $R$, but $\mathcal{D}$ has many more edges for the agents to traverse.

The second improvement of SLUG is its  better use of stigmergy and beacons. DLLG is heavily redundant, requiring both a mobile agent and a beacon at every location, and closed beacons are used only to signal the algorithm’s termination to an external observer. In SLUG, closed beacons back-propagate without requiring a mobile agent to be on top of them, meaning that we only require a beacon at every vertex to attain termination (hence ``Single-Layer'' Unlimited Gradient). Moreover, in SLUG, closed beacons are used as a signal for mobile agents not to enter already-explored regions.

The third improvement is giving mobile agents flying above closed beacons the ability to \textit{descend} (rather than ascend) the gradient of step counts, enabling them to move from fully explored areas to either unexplored areas or back to the entry point, if they cannot find any. Specifically, whereas mobile agents located at open beacons attempt to \textit{ascend} the gradient (Algorithm \ref{algorithm:SLUG}, lines 6-8), mobile agents located at closed beacons attempt to \textit{descend} the gradient (lines 9-11). This modification speeds up the exploration process. It also naturally results in redundant mobile agents going back toward an entry point at the end of the algorithm's execution, should the operator wish to recover them. We find this feature important enough to point out, despite the matter of mobile agent recovery not being a main focus of this work. 

\begin{algorithm}[!ht]
\caption{Local rule for agent $a_i$ waking up at time $t+dt$ and  executing Single-Layer Limited Gradient.}\label{algorithm:SLUG}
\begin{algorithmic}[1]
\Require $a_i$ has entered $\textbf{R}$ 
\If{$a_i$ is a ``mobile'' agent}
\If{$\xi(a_i,t)$ contains an empty cell $u$}
\State $a_i$ moves to $u$ 
\State $a_i$ sets state to ``open beacon''
\State $s_i \gets s_i + 1$ 
\ElsIf{$\xi(a_i,t)$ contains a cell $u$ with just one agent $a_j$ \textbf{and} $a_j$ is an open beacon \textbf{and} $s_j > s_i$}
\State $a_i$ moves to $u$
\State $s_i \gets s_j$ 
\ElsIf{$a_i$'s current location contains a  closed beacon \textbf{and} $\xi(a_i,t)$ contains a cell $u$ with just one agent $a_j$ \textbf{and} $a_j$ is an open or closed beacon \textbf{and} $s_j < s_i$} 
\State $a_i$ moves to $u$
\State $s_i \gets s_j$ 
\EndIf
\ElsIf{$a_i$ is an open beacon \textbf{and} every beacon $a_j \in \xi(a_i,t)$ with $s_j > s_i$ is closed \textbf{and} $\xi(a_i,t)$ does not contain an empty cell}
\State $a_i$ sets state to ``closed beacon''
\EndIf
\end{algorithmic}
\end{algorithm}

An example run of SLUG in a square grid region  is depicted in Figure \ref{fig:SLUGexamplefigure}. Simulation results suggest that SLUG terminates faster than DLLG, and that its total energy use tends to be an order of magnitude smaller (see Section \ref{section:simulations}).

\begin{figure}[ht]
    \begin{adjustwidth}{\figureWidthAdjustment}{\figureWidthAdjustment}
     \centering
     \begin{subfigure}[t]{0.231\linewidth}\hfil
         \centering
         \includegraphics[width=\linewidth]{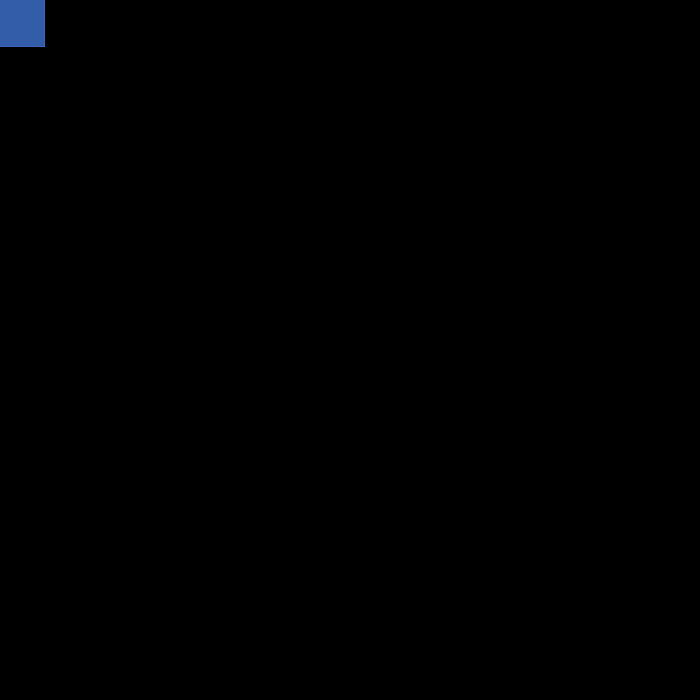}
     \end{subfigure}
     \begin{subfigure}[t]{0.231\linewidth}\hfil
         \centering
         \includegraphics[width=\linewidth]{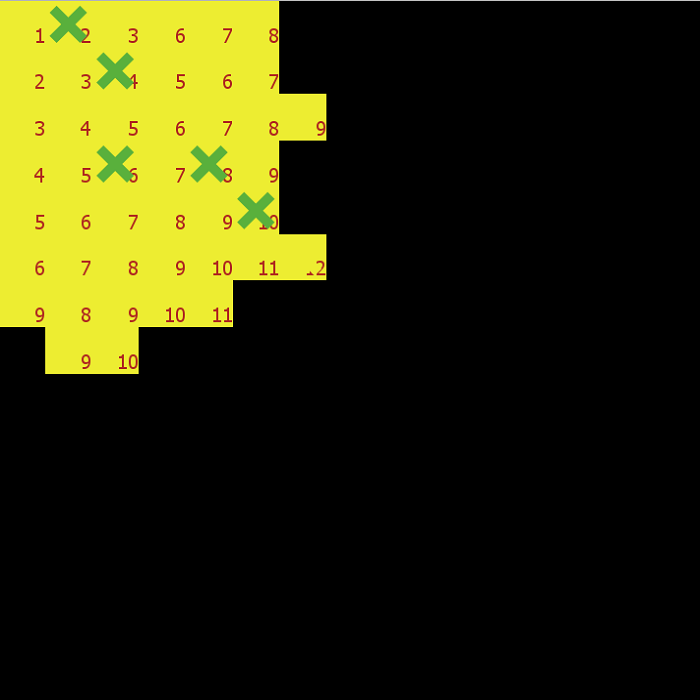}
     \end{subfigure}
     \begin{subfigure}[t]{0.231\linewidth}\hfil
         \centering
         \includegraphics[width=\linewidth]{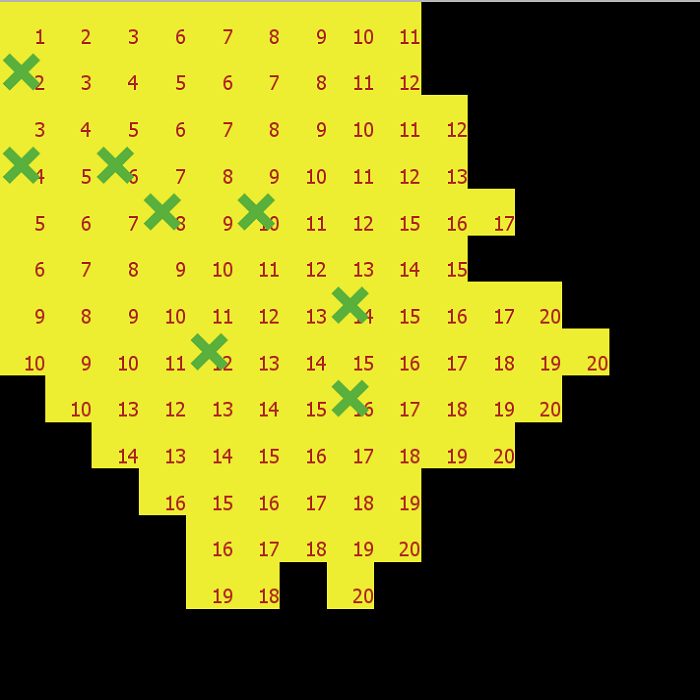}
     \end{subfigure}
     
     \begin{subfigure}[t]{0.231\linewidth}\hfil
         \centering
         \includegraphics[width=\linewidth]{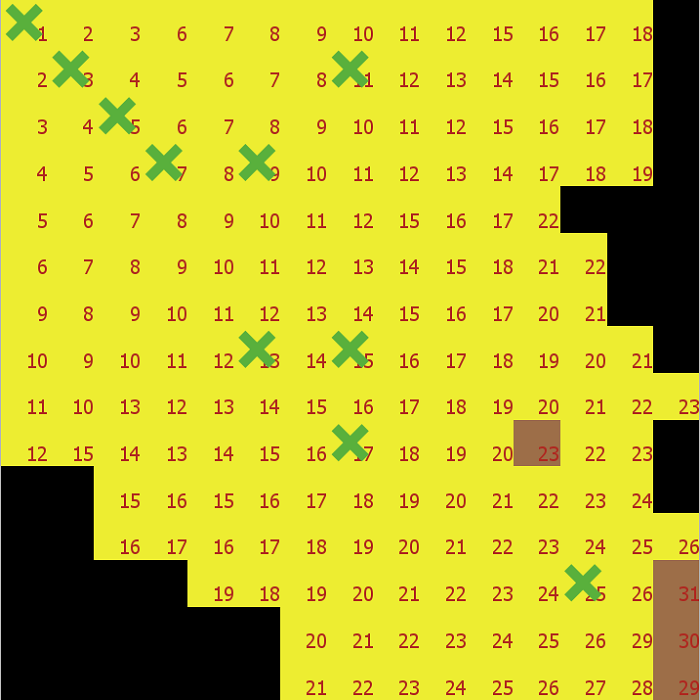}
     \end{subfigure}
     \begin{subfigure}[t]{0.231\linewidth}\hfil
         \centering
         \includegraphics[width=\linewidth]{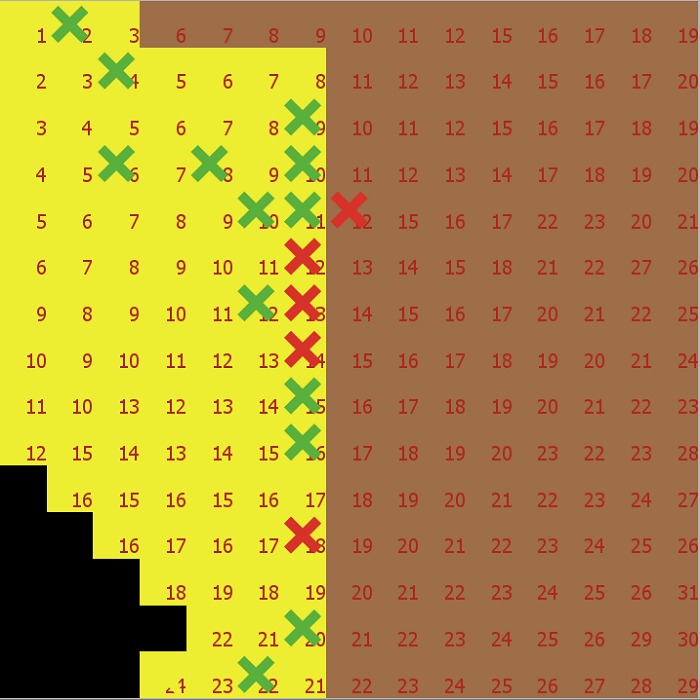}
     \end{subfigure}
     \begin{subfigure}[t]{0.231\linewidth}
         \centering
         \includegraphics[width=\linewidth]{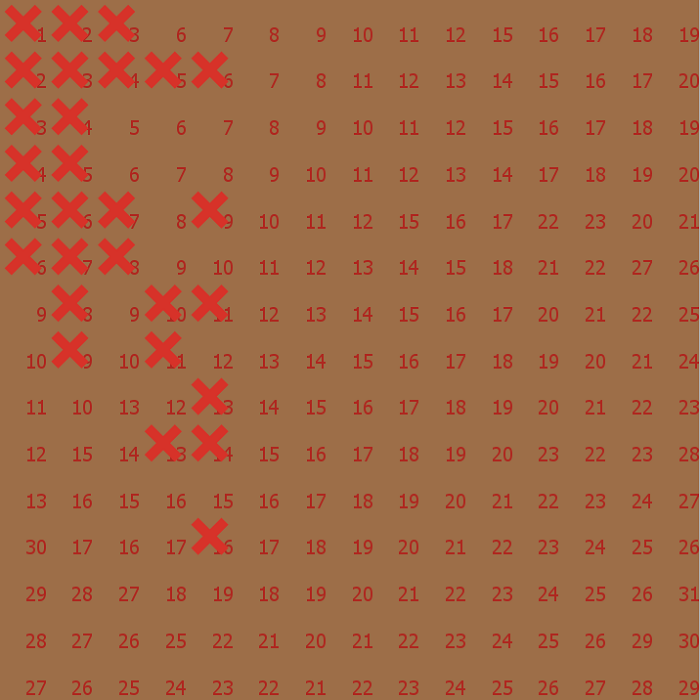}
     \end{subfigure}
     \end{adjustwidth}
     \vspace{-2mm}
    \caption{Snapshots from a run of SLUG  ordered left to right, top to bottom. Empty cells are black, obstacles are white, cells containing an open beacon are yellow and cells containing a closed beacon are brown. Beacons' step count is denoted by the red numbers. Green ``X''s are mobile agents that have never entered a closed beacon's location, and red ``X''s are mobile agents that have. Red agents descend the gradient until they find new open beacons. $\lceil T_C(\textbf{R}) \rceil = 508$.}
    \label{fig:SLUGexamplefigure}
\end{figure}
\vspace{-2mm}

\subsubsection*{``Ant-like'' recovery from crash faults and errors.}
\label{section:crashfaults}

An often-sought property of multi-robot systems is their ability to recover from crash faults and memory errors. The mobile agent layer of SLUG is provably resilient to crash faults: a mobile robot can crash and disappear from the environment at any time without affecting the algorithms' correctness (and often barely affecting termination time). Furthermore, three straight-forward changes to SLUG can be implemented to make the mobile agents provably resilient to memory corruption and positional errors:

\begin{enumerate}
    \item Instead of keeping track of its own step count, have the mobile agent $a_i$ read $s_i$ from the beacon underneath it at each wake-up.
    \item If, for some unexpected reason, a mobile agent $a_i$ does not see a beacon underneath it at wake-up, $a_i$ stays put and does nothing.
    \item A mobile agent $a_i$ that senses a location $u$ containing another mobile agent $a_j$ but no underlying beacon treats $u$ as an empty cell (Algorithm \ref{algorithm:SLUG}, lines 1-5), i.e., $a_i$ can settle underneath $a_j$ and become a beacon.
\end{enumerate}

With these changes in place, mobile agents become completely oblivious and ``ant-like,'' basing their decisions only on local sensing-based information \cite{yanivwagner2008cooperative}. Each mobile agent $a_i$ is thus stateless and its step count cannot be corrupted. Additionally, if (due to e.g. a sensing error)  $a_i$'s position changes unexpectedly, $a_i$ can always recover: if its new location contains a beacon $a_j$, $a_i$ simply continues executing SLUG as normal with its step count equalling $a_j$'s; if its new location is an empty cell, $a_i$ simply waits for a mobile agent to settle underneath it and become a beacon.


\vspace{-4mm}
\section{Simulation results and benchmarks}
\label{section:simulations}
In this section we employ numerical simulations to study the performance of SLUG and DLLG in terms of energy,  termination time, and other related parameters, and compare these algorithms to well-known algorithms from the literature. 

Recall that SLUG makes two independent conceptual improvements over DLLG: (i) it has looser traversal rules, and (ii) it makes better use of closed beacons.
It is interesting to ask what each of these improvements provides: what is the the impact of looser, or stricter traversal rules on DLLG's termination time and energy use? What is the impact of the closed beacon strategy?

We investigate these questions numerically by comparing the termination time and energy use of DLLG to SLUG and to two variants of DLLG: a variant with stricter traversal rules called \textit{Dual-Layer Tree Traversal} (DLTT), and a variant with looser traversal rules called \textit{Dual-Layer Unlimited Gradient} (DLUG). All three Dual-Layer algorithms operate under the same underlying principle of gradient climbing and dual-layer coverage, filling $\textbf{R}$ with $n$ mobile agents and $n$ closed beacons—the only difference is the underlying directed acyclic graph which the agents construct and traverse, $\mathcal{D}$. In DLTT, an algorithm first described in \cite{amir_fast_2019} and which can directly be compared to our algorithms as a benchmark, $\mathcal{D}$ is restricted to be a tree. This results in more limited traversal than DLLG, where $\mathcal{D}$ is not necessarily a tree. In DLUG, similar to SLUG, the agents are not restricted in terms of gradient steepness: an agent $a_i$ can move to any neighboring beacon whose step count is at least $s_i+1$, thus $\mathcal{D}$ is has more edges compared to DLLG (where $a_i$ can only move to a neighboring beacon whose step count is exactly $s_i+1$), and mobile agents' traversal rules are consequently less restricted. Comparing DLUG, DLLG and DLTT gives us insight into the impact of traversal rules on the energy and time costs of a dual-layer coverage algorithm. We shall show that DLUG, which has the least restrictive gradient ascent rules, outperforms DLGG and DLTT. Comparing these algorithms to SLUG gives us insight into the cumulative impact of using looser traversal rules \textit{and} the closed beacon strategy. We find that SLUG  significantly outperforms DLLG, DLTT and DLUG. 

Our numerical simulations were run on the regions depicted in Figure \ref{fig:simulationexampleregions}. The top three regions are simple regions that can be generalized to different sizes, and through which the effect of region size on termination time and energy use can be  investigated. The three regions depicted in the bottom row were replicated from \cite{ferranti_brick_2007} as  examples of complex regions that more closely resemble real-life scenarios. Each data point in the figures of this section is the average of data from 50 runs of an algorithm on one of these regions (when simulating the time model we set $M=100$).

\begin{figure}[ht]
    \begin{adjustwidth}{\figureWidthAdjustment}{\figureWidthAdjustment}
     \centering
     \begin{subfigure}[t]{0.22\linewidth}\hfil
         \centering
         \includegraphics[width=\linewidth]{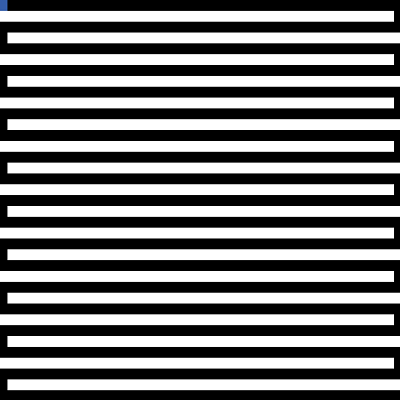}
         \subcaption{}
     \end{subfigure}
     \begin{subfigure}[t]{0.22\linewidth}\hfil
         \centering
         \includegraphics[width=\linewidth]{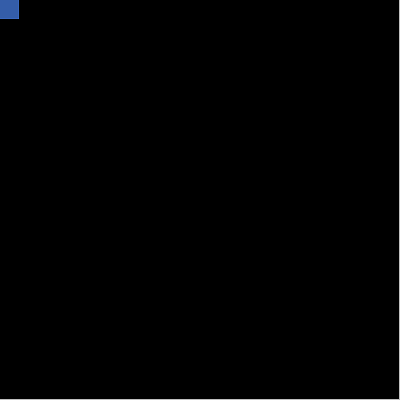}
         \subcaption{}
     \end{subfigure}
     \begin{subfigure}[t]{0.22\linewidth}\hfil
         \centering
         \includegraphics[width=\linewidth]{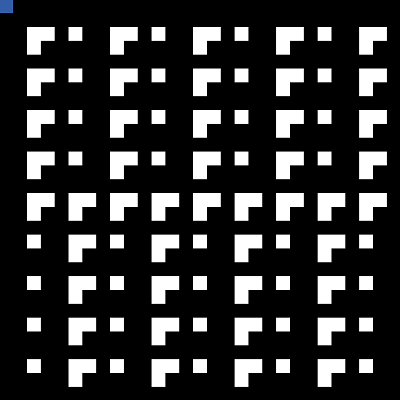}
         \subcaption{}
     \end{subfigure}
     
     \begin{subfigure}[t]{0.22\linewidth}\hfil
         \centering
         \includegraphics[width=\linewidth]{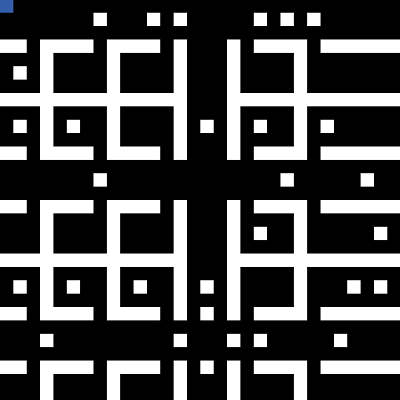}
         \subcaption{}
     \end{subfigure}
     \begin{subfigure}[t]{0.22\linewidth}\hfil
         \centering
         \includegraphics[width=\linewidth]{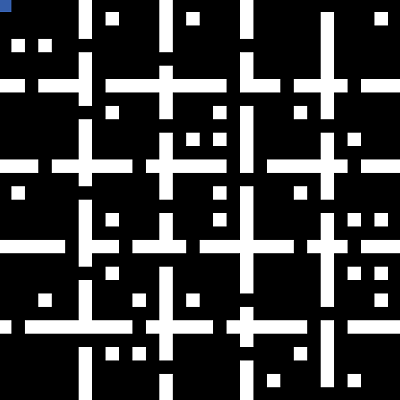}
         \subcaption{}
     \end{subfigure}
     \begin{subfigure}[t]{0.22\linewidth}
         \centering
         \includegraphics[width=\linewidth]{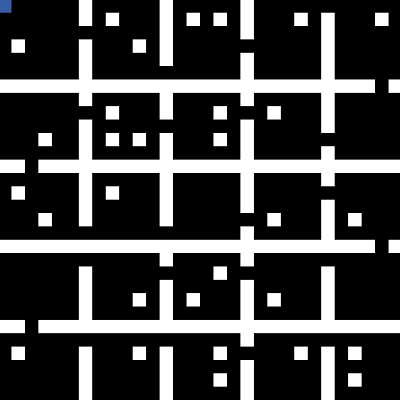}
         \subcaption{}
     \end{subfigure}
     \end{adjustwidth}
     \vspace{-2mm}
    \caption{Regions used in simulation experiments. Obstacles are marked in white and  empty cells are marked in black. The entry point is at the upper-left corner, marked in blue. Region labels are as follows: (a) linear (b) square (c) saw tooth (d) complex 1 (e) complex 2 (f) complex 3.
    Regions (a), (b) and (c) can be generalized to different sizes - e.g., we run simulations on linear regions of different sizes.}
    \vspace{-4.5mm}
    \label{fig:simulationexampleregions}
\end{figure}

The impact of $\Delta T$ on DLLG and SLUG's termination time step $\lceil T_C(\textbf{R}) \rceil$ is investigated in Figure \ref{fig:DLLGcoverageDeltaT}, (a)  for the 6 types of regions depicted in Figure \ref{fig:simulationexampleregions}. We see that for every value of $\Delta T$ and every region, SLUG's termination time is significantly lower than DLLG's. Figure \ref{fig:DLLGcoverageDeltaT}, (a) also confirms the termination time bounds of Theorem \ref{theorem:dllgterminationtime} for DLLG: as claimed, we clearly see that $\lceil T_C(\textbf{R}) \rceil/n \approxeq 2\Delta T$ in the plot of Figure \ref{fig:DLLGcoverageDeltaT}, (a) corresponding to DLLG. Figure \ref{fig:DLLGcoverageDeltaT}, (b) investigates the number of agents present in a region at termination time when using SLUG and DLLG. By construction, every region always contains $2n$ agents at DLLG's termination, but when it comes to SLUG, the ratio of agents to region size quickly converges to $1$ as $\Delta T$ goes to infinity. In other words, as $\Delta T$ grows, there are less redundant agents in the environment. \textit{Hence, by increasing $\Delta T$ in SLUG, we can reduce the number of redundant agents at the cost of larger makespan.}

Figure \ref{fig:assorteDLLGDLUGDLTT} depicts the effect of region size on the termination time and total energy use of SLUG, DLTT, DLLG and DLUG when $\Delta T = 1$ (Figure \ref{fig:assorteDLLGDLUGDLTT}, (a) and (c)) and when $\Delta T = 2$ (Figure \ref{fig:assorteDLLGDLUGDLTT}, (b) and (d)). We see that SLUG significantly outperforms DLTT, DLGG and DLUG in terms of energy use and termination time. Additionally, several facts about the impact of traversal rules can be inferred from the figure: first, looser traversal rules confer a benefit on energy use:  $E_{total}$ is reduced in DLUG compared to DLTT and DLLG. Next, when $\Delta T \geq 2$, termination time is unaffected by the looseness or strictness of traversal rules. However, when $\Delta T = 1$, looser traversal rules improve termination time: of the three dual-layer algorithms, DLUG's termination time is lowest, followed by DLLG.

\begin{figure}[ht]
    \begin{adjustwidth}{\figureWidthAdjustment}{\figureWidthAdjustment}
     \centering
     \begin{subfigure}[t]{0.455\linewidth}\hfil
         \centering
         \includegraphics[width=\linewidth]{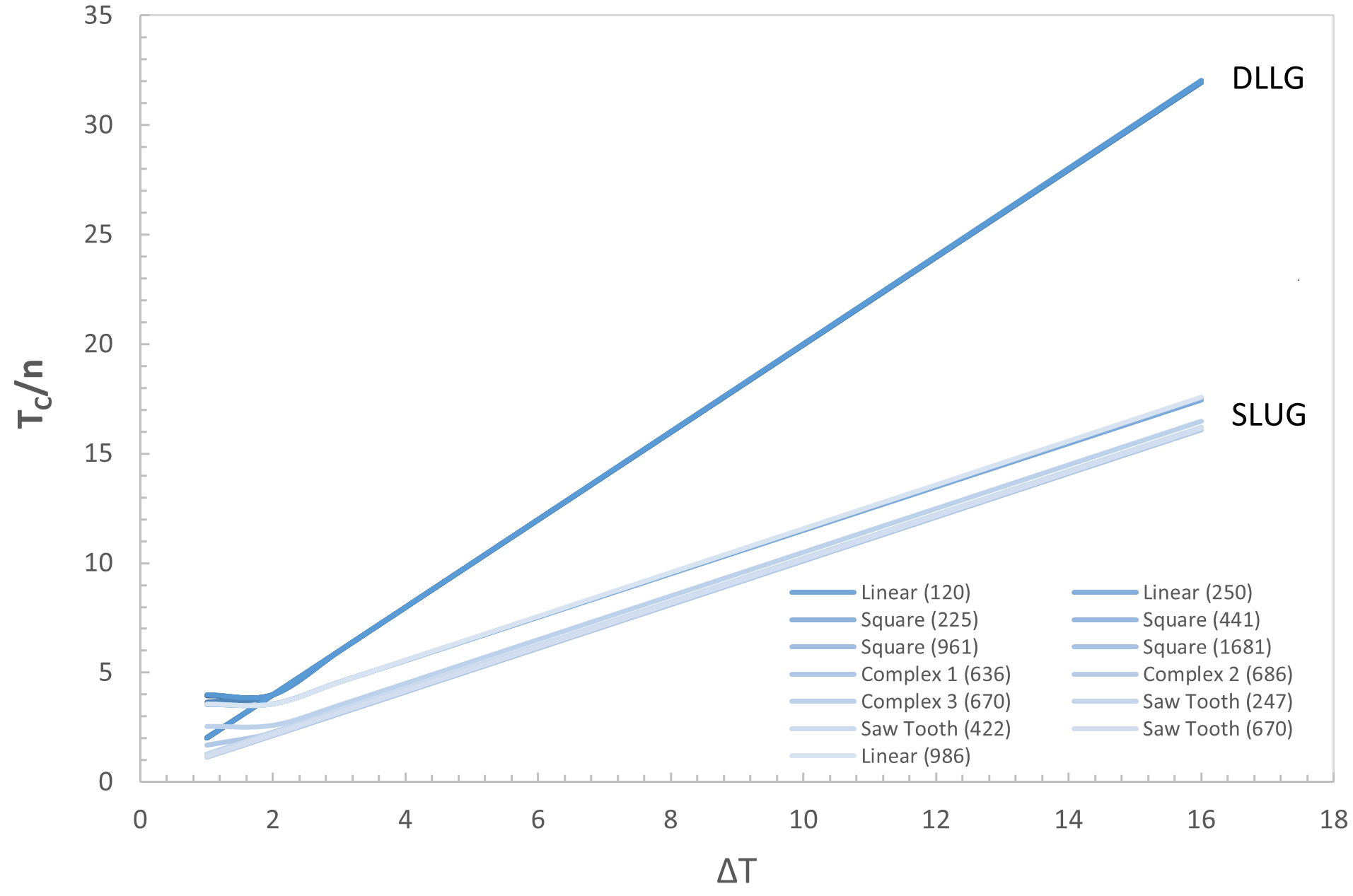}
         \subcaption{}
     \end{subfigure}
     \begin{subfigure}[t]{0.455\linewidth}\hfil
         \centering
         \includegraphics[width=\linewidth]{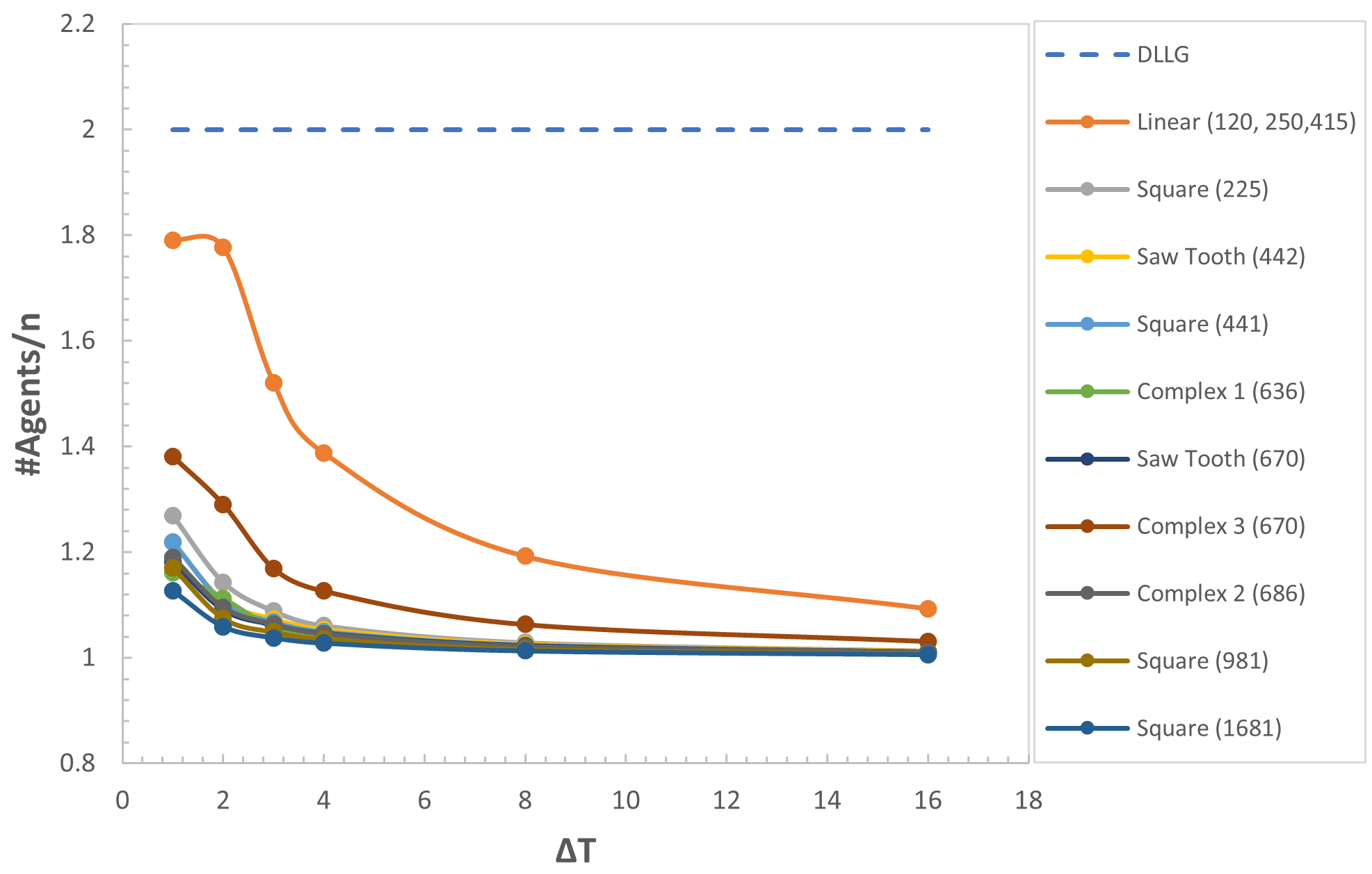}
         \subcaption{}
     \end{subfigure}
     \end{adjustwidth}
     \vspace{-2mm}
    \caption{(a) Each line plots $\lceil T_C(\textbf{r}) \rceil/n$ (termination time step divided by region size) against $\Delta T$ in simulations of DLLG or SLUG on one of the regions of Figure \ref{fig:simulationexampleregions}. The color of a line determines the region, as detailed in the legend. The parentheses in the legend indicate region size.  We plot two lines of the same color per region - one corresponding to simulations of DLLG on that region, and one corresponding to simulations of SLUG. Lines corresponding to SLUG's termination time on a region are strictly below those corresponding to DLLG, indicating that SLUG outperforms DLLG. (b) Each line plots the number of agents present in a region at termination time divided by the size of the region, against $\Delta T$.}
    \label{fig:DLLGcoverageDeltaT}
\end{figure}
\vspace{-2mm}

\begin{figure}[ht]
    \begin{adjustwidth}{\figureWidthAdjustment}{\figureWidthAdjustment}
     \centering
     \begin{subfigure}[t]{0.43\linewidth}\hfil
         \centering
         \includegraphics[width=\linewidth]{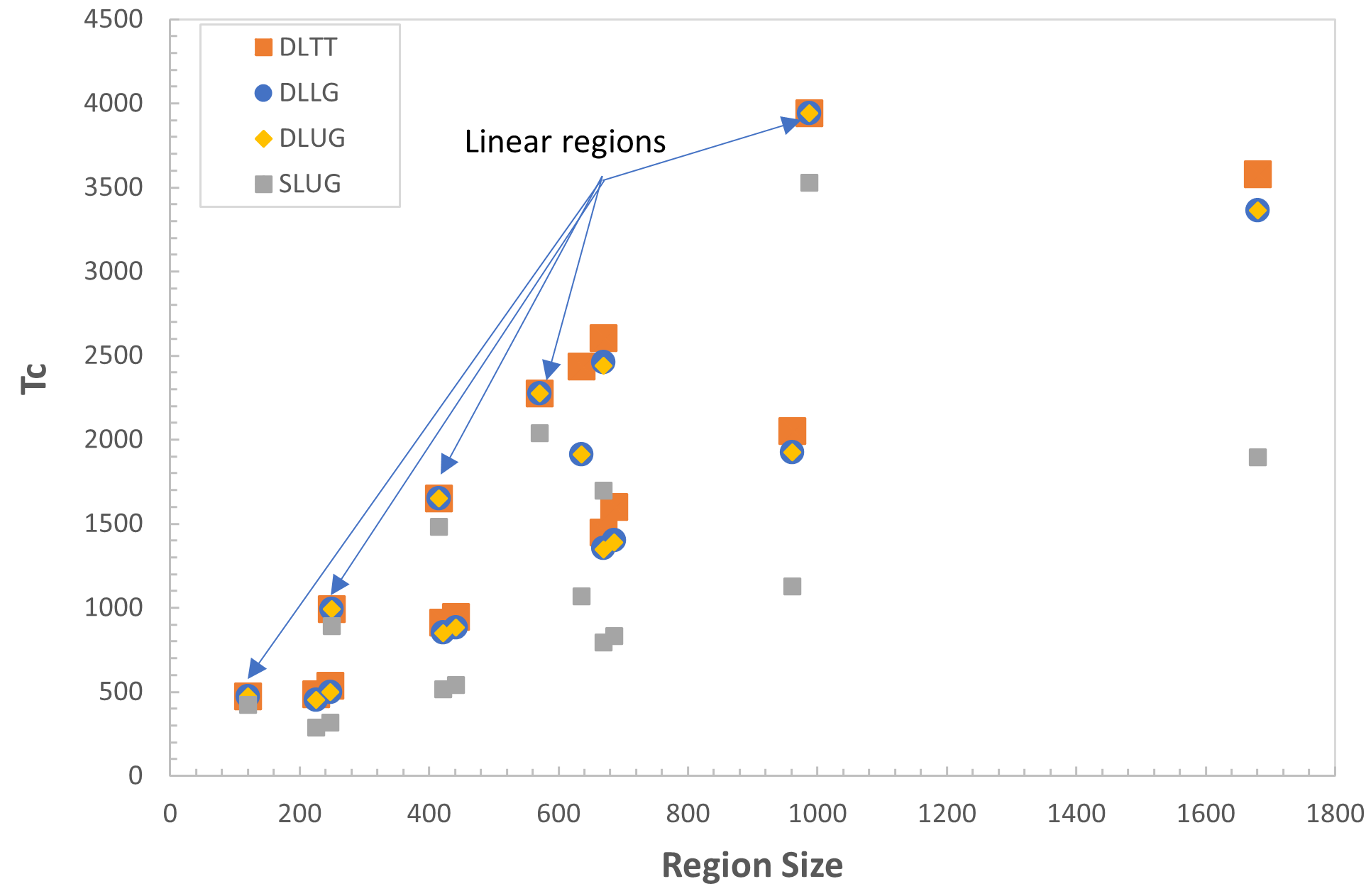}
         \subcaption{}
     \end{subfigure}
     \begin{subfigure}[t]{0.43\linewidth}\hfil
         \centering
         \includegraphics[width=\linewidth]{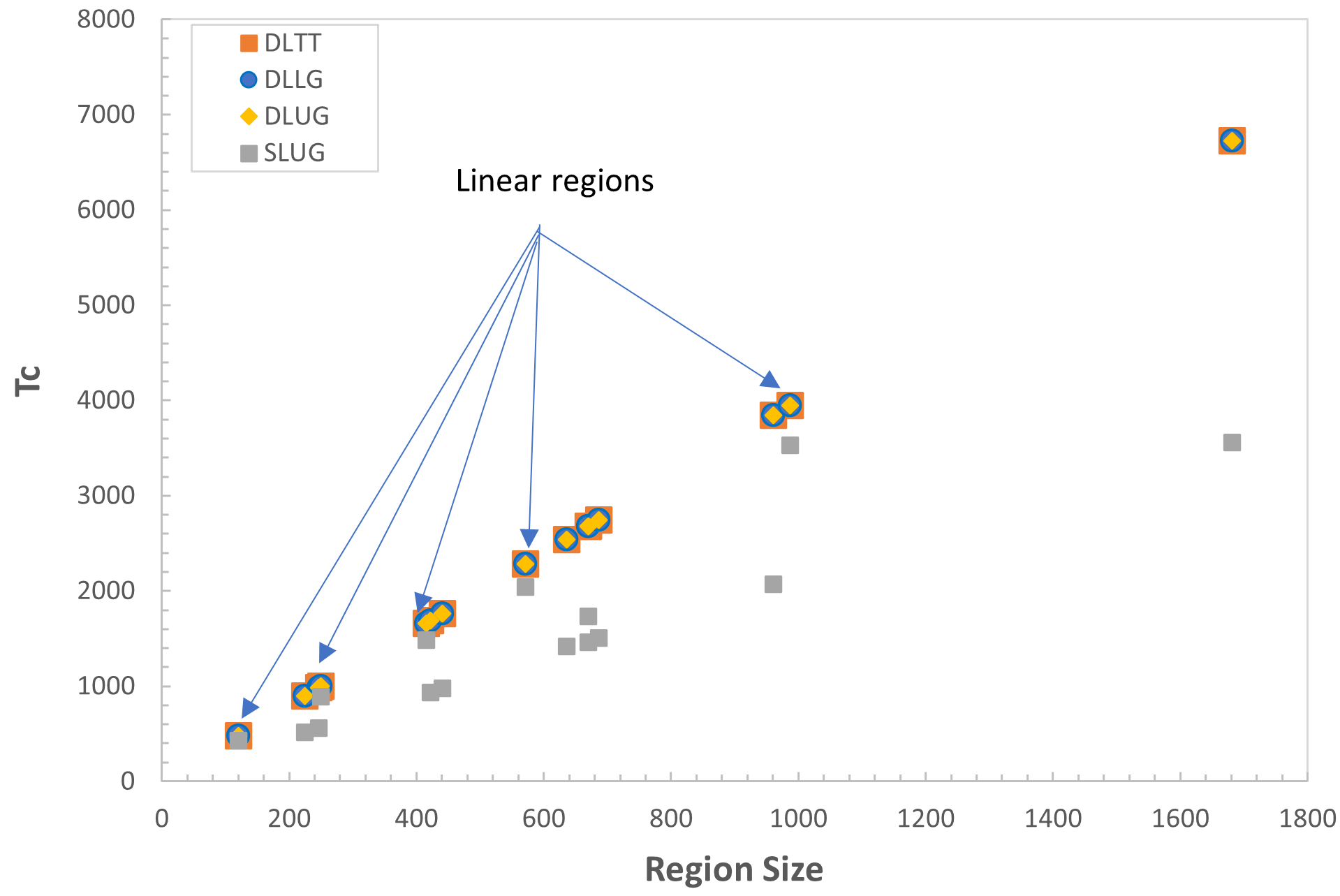}
         \subcaption{}
     \end{subfigure}

     \begin{subfigure}[t]{0.43\linewidth}\hfil
         \centering
         \includegraphics[width=\linewidth]{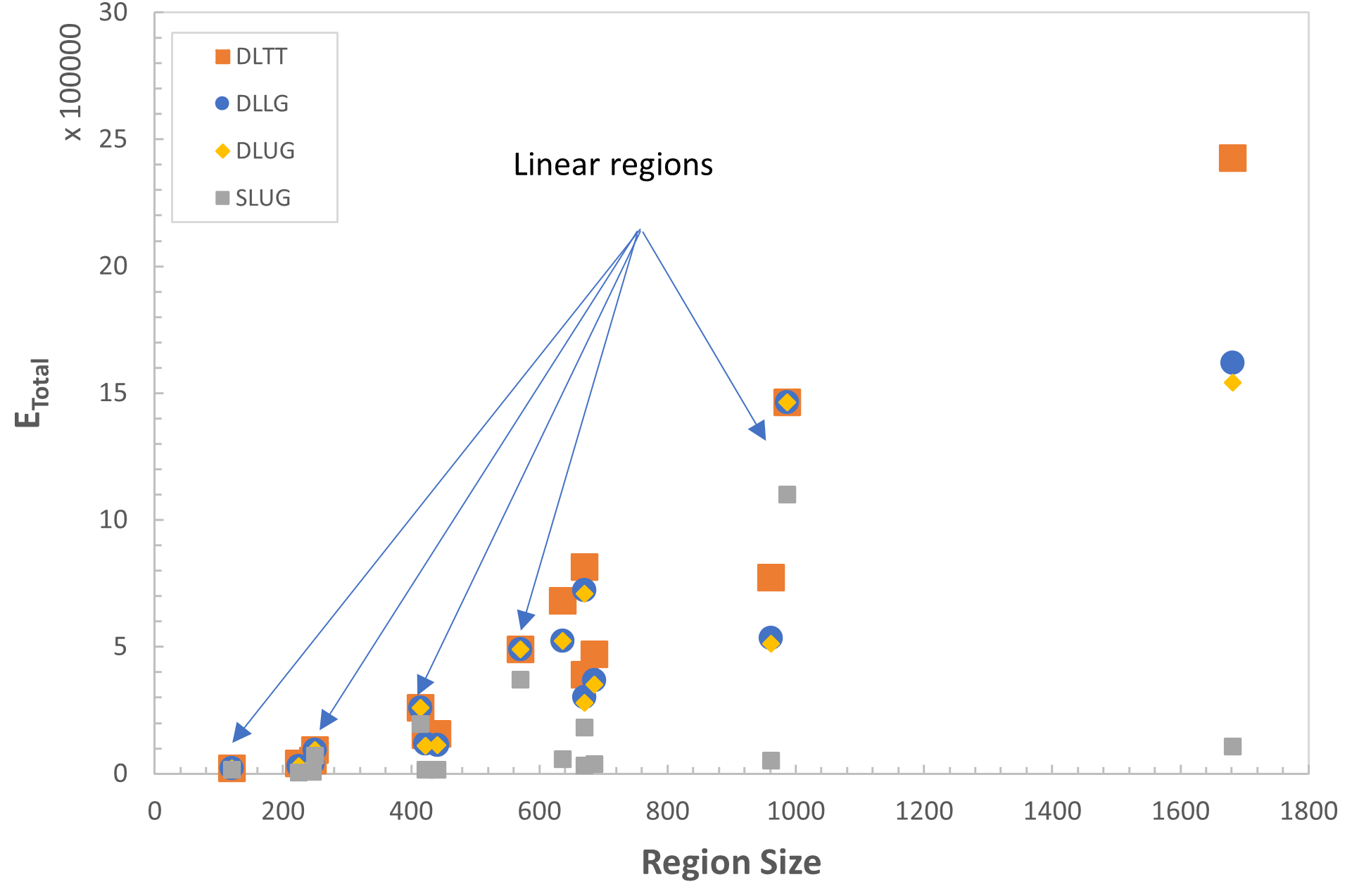}
         \subcaption{}
     \end{subfigure}
     \begin{subfigure}[t]{0.43\linewidth}
         \centering
         \includegraphics[width=\linewidth]{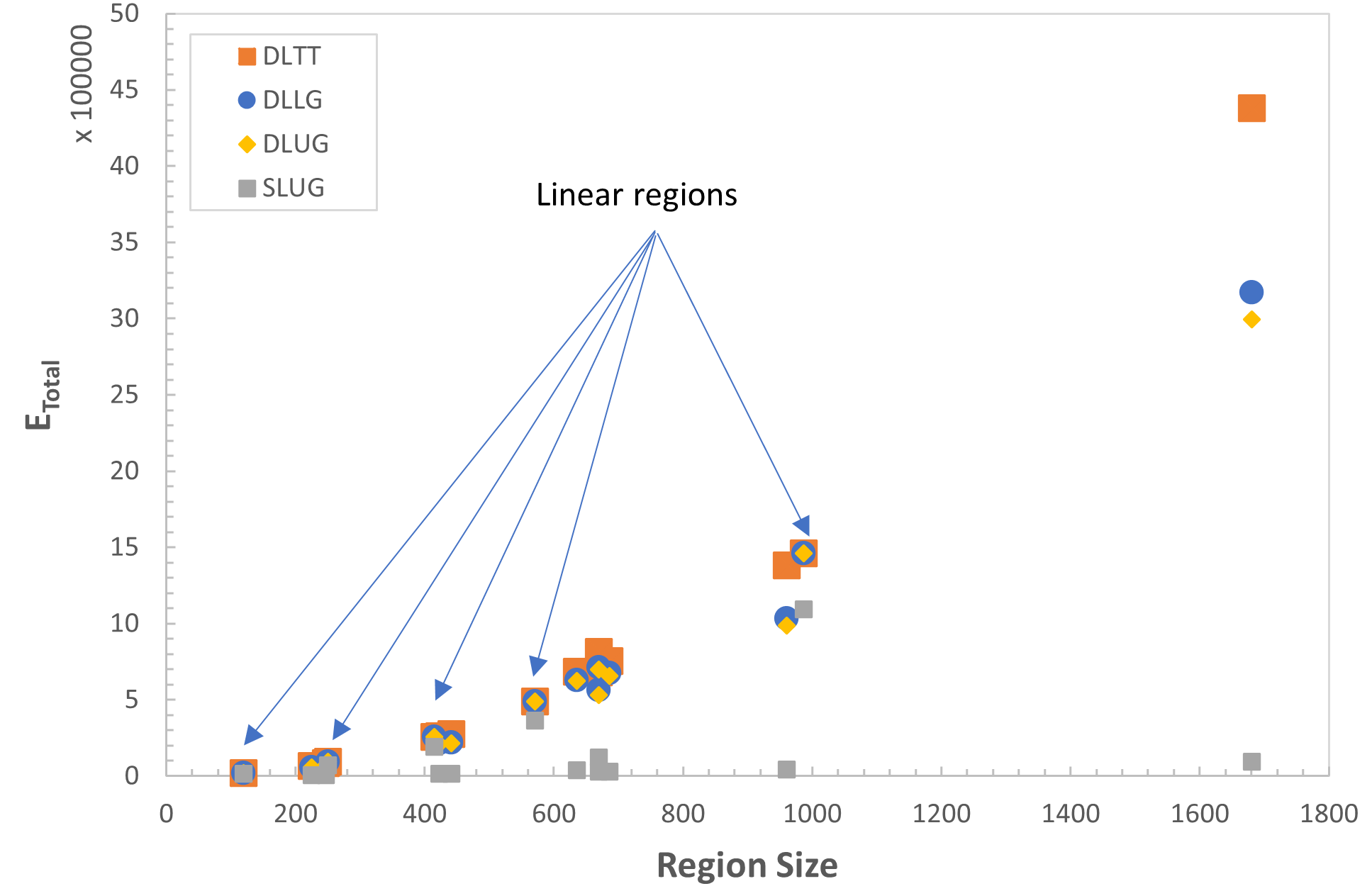}
         \subcaption{}
     \end{subfigure}
     \end{adjustwidth}
     \vspace{-2mm}
    \caption{(a) Each data point denotes the average termination time step $\lceil T_C(\textbf{r}) \rceil$  of DLTT (orange square), DLLG (blue cirle), DLUG (yellow diamond) or SLUG (gray square) on one of the region types in Figure \ref{fig:simulationexampleregions} plotted against region size, assuming $\Delta T = 1$. 
    Linear regions  of various sizes  (Figure \ref{fig:simulationexampleregions}, (a)) had the worst termination times, and data points representing these regions are marked by blue arrows. (b) Each data point denotes average termination time step when $\Delta T = 2$. 
    (c) Each data point denotes total energy use step when $\Delta T = 1$. (d) Each data point denotes total energy use when $\Delta T = 2$.}
    \label{fig:assorteDLLGDLUGDLTT}
    \vspace{-2mm}
\end{figure}


\subsubsection*{Benchmarks from the literature.} We compared SLUG to the well-known Depth-First and Breadth-First Leader-Follower coverage algorithms of Hsiang et al. \cite{hsiang_algorithms_2004} in terms of termination time and energy use. Hsiang's algorithms assume synchronous time, a greater sensing range than ours, and direct radio-based inter-agent communication.
Due to these differing assumptions, an exact one-to-one comparison of our work and \cite{hsiang_algorithms_2004} is not possible--but we can nevertheless draw some conclusions.
Figure \ref{fig:SLUGhsiangcomparison}, (a) compares the termination time of SLUG to Hsiang et al.'s algorithms on square grid regions. Assuming $\Delta T = 1$, SLUG's termination time is significantly lower than BFLF or DFLF's, whereas when $\Delta T = 2$, its termination time is slightly higher. Figure \ref{fig:SLUGhsiangcomparison}, (b) compares total energy use, showing that SLUG's energy use (in both the $\Delta T = 1$ and $\Delta T = 2$ cases) is significantly lower than either of Hsiang's algorithms (by nearly an order of magnitude on a 41x41 square grid). Figure \ref{fig:SLUGhsiangcomparison}, (c) shows that SLUG's maximum energy use (Definition \ref{def:maximumenergy}) is significantly lower than BFLF and DFLF's as well. These findings suggest that the beacon-based approach results in large energy savings compared to BFLF and DFLF, despite the fact that (unlike Hsiang et al.) SLUG inserts more than $n$ robots into the region, thus inevitably ``wastes'' energy on redundant mobile robots that are not needed to complete the coverage.

Another interesting metric that Hsiang et al. measures is the \textit{total travel}, i.e., the total number of movements the robots make. We measured total travel (not depicted in the Figure) and found that SLUG significantly outperforms DFLF, and is comparable to BFLF in terms of total travel on all of our test regions.

\begin{figure}[ht]
    \begin{adjustwidth}{-0.9cm}{-0.9cm}
     \centering
     \begin{subfigure}[t]{0.41\linewidth}\hfil
         \centering
         \includegraphics[width=\linewidth]{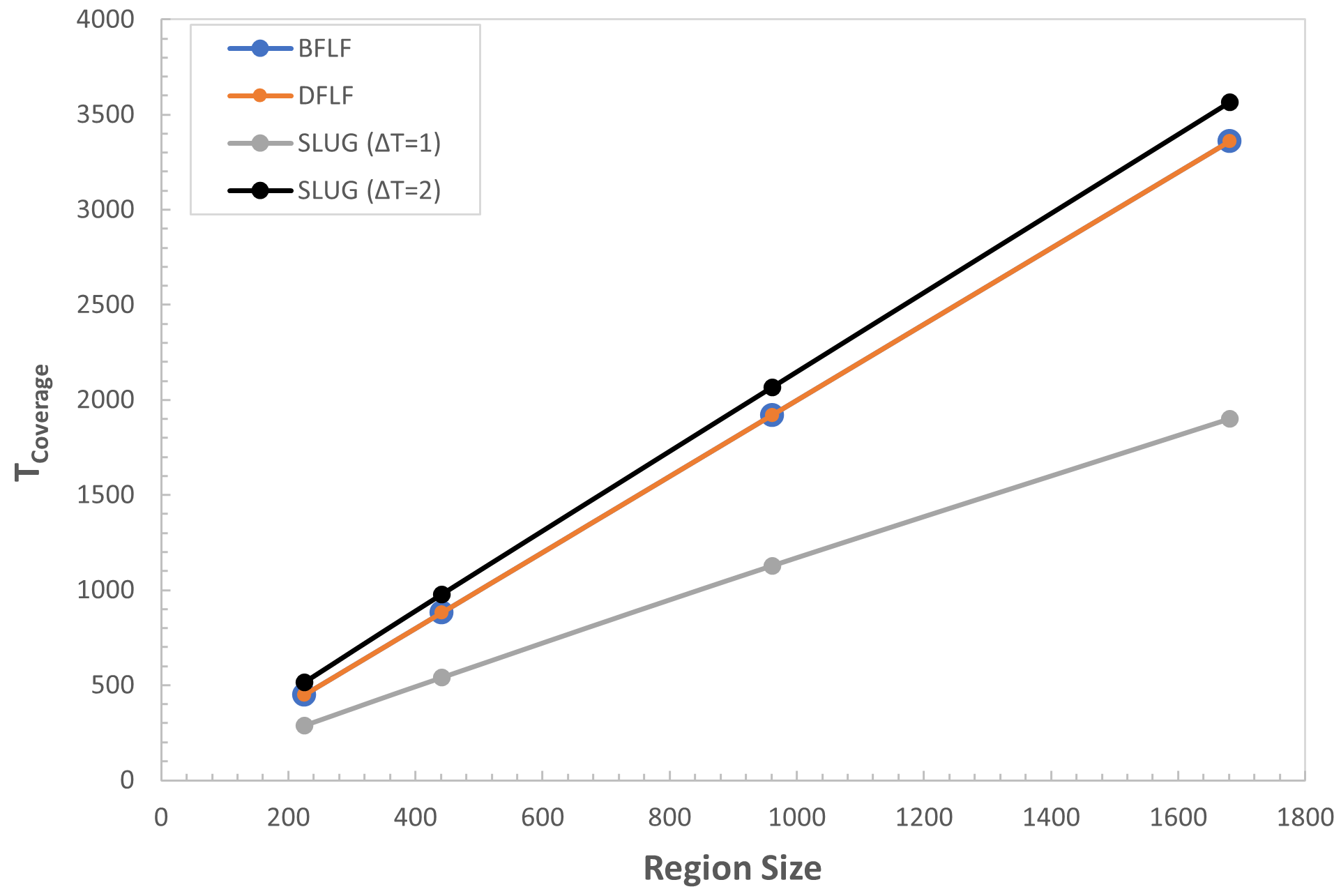}
         \subcaption{}
     \end{subfigure}
     \begin{subfigure}[t]{0.41\linewidth}\hfil
         \centering
         \includegraphics[width=\linewidth]{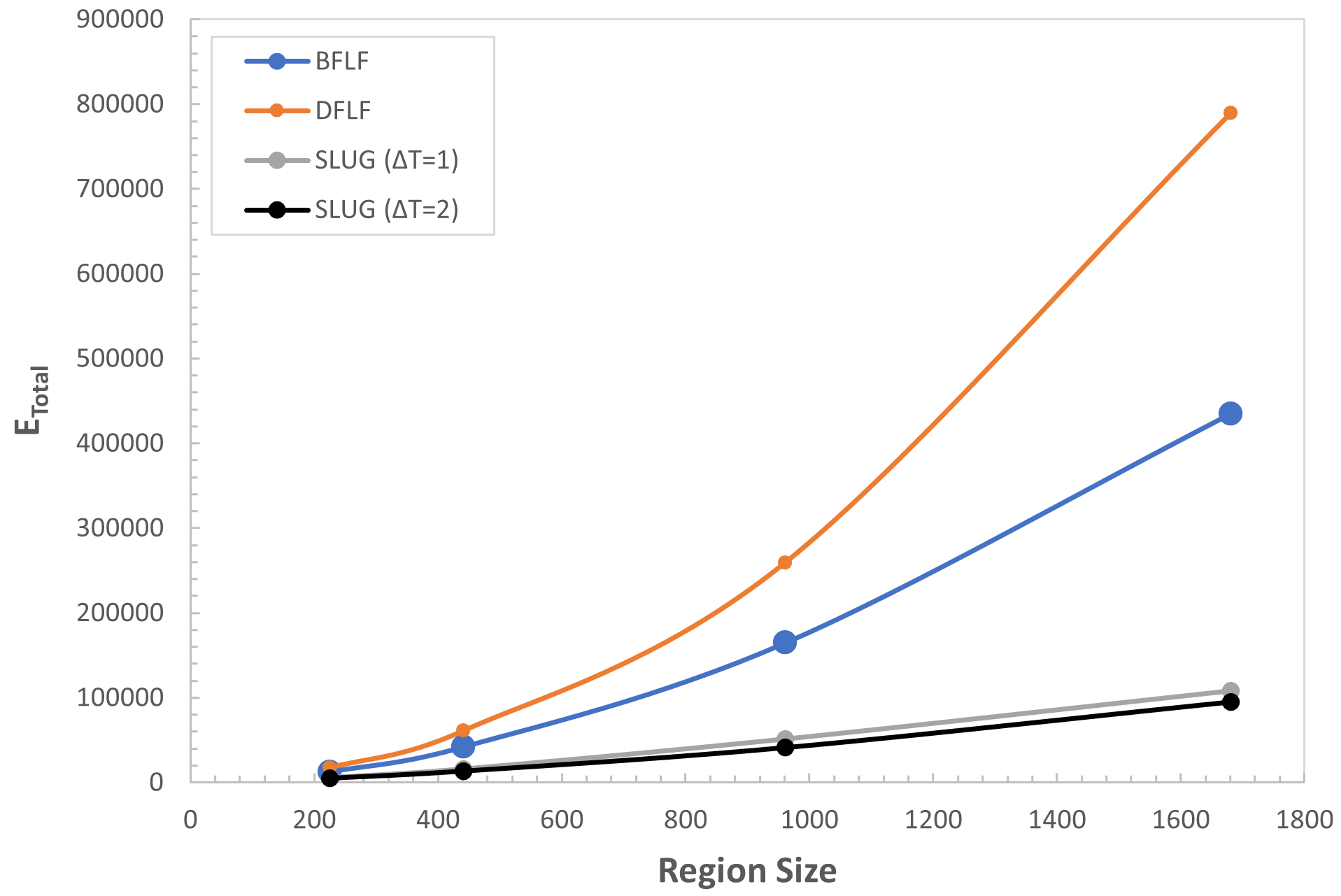}
         \subcaption{}
     \end{subfigure}

     \begin{subfigure}[t]{0.41\linewidth}\hfil
         \centering
         \includegraphics[width=\linewidth]{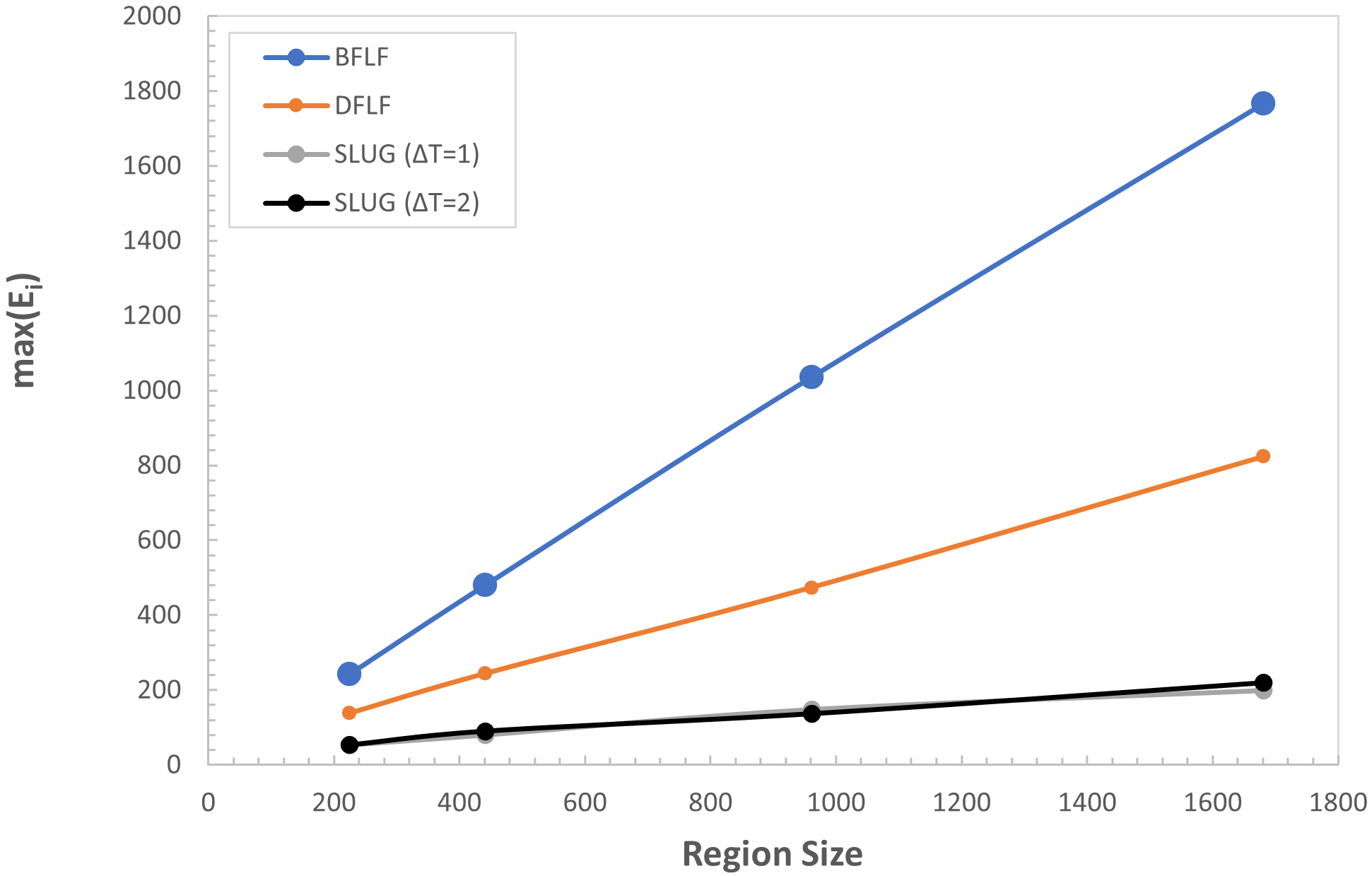}
         \subcaption{}
     \end{subfigure}
     \end{adjustwidth}
     \vspace{-2mm}
    \caption{(a) The termination time of SLUG and Hsiang et al.'s BFLF and DFLF algorithms on square grid regions of various sizes (Figure \ref{fig:simulationexampleregions}, (b)). (b) The total energy use of the algorithms on square grid regions. (c) The maximum individual energy (Definition \ref{def:maximumenergy}) use of the algorithms  on square grid regions.} 
    \label{fig:SLUGhsiangcomparison}
\end{figure}
\vspace{-3mm}


\section{Discussion}

In this work we investigated the problem of covering unknown indoor environments with a swarm of aerial robots. To address the energy and communication challenges that arise in this setting, we proposed algorithms that draw upon the meta-concept of ``stigmergy'' - communication via the environment. In our algorithms, robots that settle inside the region become part of the environment in the form of beacons that guide mobile robots to as-yet unexplored locations. We described two algorithms based on this concept, DLLG and SLUG, and studied their performance through formal mathematical analysis (of DLLG) and through simulation results. SLUG attains large energy savings over algorithms that appeared in previous work (\cite{amir_fast_2019} and \cite{hsiang_algorithms_2004}) and is competitive in terms of termination time and other metrics, hinting at the potential of beacon-based approaches. 

We view our results as a preliminary exploration of the advantages of beacon-based, dual-layer approaches in multi-robot deployment problems. Several extensions of our work readily come to mind. First, since we wanted to be able to derive rigorous mathematical results, we studied an idealized robot model in which movements are instantaneous and precise. In a real life implementation of beacon coverage, however, motion, sensing, and localization pose technical challenges that need to be addressed. In future work, we are interested in experimenting with  physical implementations, and are currently evaluating various platforms (e.g., Crazyflie 2.1). Next, we've shown that SLUG, by construction, directs  redundant mobile agents that are left over after coverage is complete toward one of the entry points, where they can be recovered. However, we did not formally model this process in this paper nor consider its impact on energy. We believe much further research should be done regarding the time and energy costs of agent recovery upon mission termination, both in this and broader contexts. Finally, given that we've shown several beneficial effects of  agent-beacon duality in swarm-robotic coverage tasks, we may naturally ask what other swarm robotics-related problems can benefit from this concept, and in what other ways this idea can be employed.



\bibliography{biblio.bib}
\bibliographystyle{plain}
\end{document}